\documentclass[journal,10pt]{IEEEtran}

\usepackage{mathtools}
\usepackage{amssymb}
\usepackage{amsthm}
\usepackage{siunitx}
\usepackage[utf8]{inputenc}
\usepackage{ifthen}
\usepackage{microtype}
\usepackage{bm}
\usepackage{pgfplotstable}
\usepackage{cite}
\usepackage{hyperref}
\usepackage[acronym,shortcuts]{glossaries}
\usepackage{cleveref}
\usepackage{algorithm}
\usepackage{algorithmic}

\newcommand{\mbA}{\bm{A}}
\newcommand{\mbB}{\bm{B}}
\newcommand{\mbC}{\bm{C}}
\newcommand{\mbD}{\bm{D}}

\newcommand{\mbF}{\bm{F}}
\newcommand{\mbG}{\bm{G}}
\newcommand{\mbH}{\bm{H}}
\newcommand{\mbI}{\bm{I}}

\newcommand{\mbN}{\bm{N}}

\newcommand{\mbP}{\bm{P}}

\newcommand{\mbX}{\bm{X}}
\newcommand{\mbY}{\bm{Y}}

\newcommand{\mba}{\bm{a}}

\newcommand{\mbc}{\bm{c}}
\newcommand{\mbd}{\bm{d}}

\newcommand{\mbh}{\bm{h}}

\newcommand{\mbn}{\bm{n}}

\newcommand{\mbs}{\bm{s}}

\newcommand{\mbx}{\bm{x}}
\newcommand{\mby}{\bm{y}}

\newcommand{\mbSigma}{{\bm{\Sigma}}}

\newcommand{\mbdelta}{{\bm{\delta}}}

\newcommand{\mbmu}{{\bm{\mu}}}

\newcommand{\mbzero}{{\bm{0}}}

\newcommand{\calB}{\mathcal{B}}
\newcommand{\calC}{\mathcal{C}}

\newcommand{\calK}{\mathcal{K}}

\newcommand{\calM}{\mathcal{M}}
\newcommand{\calN}{\mathcal{N}}
\newcommand{\calO}{\mathcal{O}}

\Crefname{figure}{Fig.}{Figs.}

\hypersetup{
    colorlinks = true,
    linkbordercolor = {white},
    citecolor = {black},
    linkcolor = {black},
    urlcolor = {black}
}

\pgfplotsset{
	discard if/.style 2 args={
        x filter/.append code={
            \edef\tempa{\thisrow{#1}}
            \edef\tempb{#2}
            \ifx\tempa\tempb
                
            \fi
        }
    },
    discard if not/.style 2 args={
        x filter/.append code={
            \edef\tempa{\thisrow{#1}}
            \edef\tempb{#2}
            \ifx\tempa\tempb
            \else
                
            \fi
        }
    }
}
\pgfplotsset{tick label style={font=\small},label style={font=\small},legend style={font=\footnotesize}}
\pgfplotsset{compat=1.15}

\usetikzlibrary{patterns}

\DeclareMathOperator{\diag}{diag}
\DeclareMathOperator{\expec}{E}
\DeclareMathOperator{\vect}{vec}

\newcommand*{\C}{\mathbb{C}}
\newcommand*{\N}{\mathbb{N}}
\newcommand*{\R}{\mathbb{R}}

\newcommand{\Brv}[1]{\bm #1}
\newcommand{\covhi}{\mbC_i}
\newcommand{\covhk}{\mbC_k}

\newcommand{\herm}{{\operatorname{H}}}
\newcommand{\hest}{\hat{\mbh}^{(K)}}
\newcommand{\hhat}{\hat{\mbh}}
\newcommand{\Krx}{K_{\text{rx}}}
\newcommand{\Ktx}{K_{\text{tx}}}
\newcommand{\meanhi}{\mbmu_i}
\newcommand{\meanhk}{\mbmu_k}
\newcommand{\Nrx}{N_{\text{rx}}}
\newcommand{\Ntx}{N_{\text{tx}}}
\newcommand{\tp}{{\operatorname{T}}}

\newacronym{amp}{AMP}{approximate message passing}
\newacronym{cae}{CAE}{concrete autoencoder}
\newacronym{cnn}{CNN}{convolutional neural network}
\newacronym{cme}{CME}{conditional mean estimator}
\newacronym{cs}{CS}{compressive sensing}
\newacronym{dft}{DFT}{discrete Fourier transform}
\newacronym{em}{EM}{expectation-maximization}
\newacronym{fdd}{FDD}{frequency division duplex}
\newacronym{gmm}{GMM}{Gaussian mixture model}
\newacronym{lmmse}{LMMSE}{linear minimum mean square error}
\newacronym{los}{LOS}{line of sight}
\newacronym{ls}{LS}{least squares}
\newacronym{mimo}{MIMO}{multiple-input multiple-output}
\newacronym{mpc}{MPC}{multi-path component}
\newacronym{mse}{MSE}{mean square error}
\newacronym{nlos}{NLOS}{non-line of sight}
\newacronym{o2i}{O2I}{outdoor-to-indoor}
\newacronym{omp}{OMP}{orthogonal matching pursuit}
\newacronym{pdf}{PDF}{probability density function}
\newacronym{simo}{SIMO}{single-input multiple-output}
\newacronym{siso}{SISO}{single-input single-output}
\newacronym{snr}{SNR}{signal-to-noise ratio}
\newacronym{ula}{ULA}{uniform linear array}

\newtheorem{lemma}{Lemma}
\newtheorem{theorem}{Theorem}

\definecolor{myblack}{RGB}{70,70,70}
\definecolor{myblue}{RGB}{65,105,225}
\definecolor{mygreen}{RGB}{0,139,139}
\definecolor{myorange}{RGB}{255,150,0}
\definecolor{myred}{RGB}{255,69,0}
\definecolor{mylila}{RGB}{153,50,204}
\definecolor{TUMMediumGray}{RGB}{128,128,128}
\definecolor{TUMBeamerBlue}{rgb}{0.00,0.60,1.00}
\definecolor{TUMBeamerYellow}{rgb}{1.00,0.71,0.00}
\definecolor{TUMBeamerOrange}{rgb}{1.00,0.50,0.00}
\definecolor{TUMBeamerRed}{rgb}{0.90,0.20,0.09}
\definecolor{TUMBeamerGreen}{rgb}{0.57,0.67,0.42}

\newcommand{\legendAmp}         {\footnotesize AMP}
\newcommand{\legendCae}         {\footnotesize CAE}
\newcommand{\legendChannelnet}  {\footnotesize ChannelNet}
\newcommand{\legendCnn}         {\footnotesize CNN}
\newcommand{\legendGenielmmse}  {\footnotesize gen. LMMSE}
\newcommand{\legendGlobalcov}   {\footnotesize sample cov.}
\newcommand{\legendGmm}         {\footnotesize GMM}
\newcommand{\legendGmmDiag}     {\footnotesize circ. GMM}
\newcommand{\legendGmmKron}     {\footnotesize Kron. GMM}
\newcommand{\legendLs}          {\footnotesize LS}
\newcommand{\legendOmp}         {\footnotesize gen. OMP}

\newcommand{\plotwidth}                         {0.95\columnwidth}

\newcommand{\plotheightSimoOnepath}             {0.65\columnwidth}
\newcommand{\plotheightSimoThreepath}           {0.65\columnwidth}
\newcommand{\plotheightSimoQuadriga}            {0.65\columnwidth}
\newcommand{\plotheightSimoComponents}          {0.45\columnwidth}
\newcommand{\plotheightMimo}                    {0.55\columnwidth}
\newcommand{\plotheightMimoTrainsamples}        {0.4\columnwidth}
\newcommand{\plotheightWidebandComponents}      {0.55\columnwidth}
\newcommand{\plotheightWidebandFixedSpeed}      {0.7\columnwidth}
\newcommand{\plotheightWidebandVariousSpeed}    {0.7\columnwidth}

\newcommand{\lineWidth}{1.0pt}
\newcommand{\marksize}{1.6pt}
\tikzset{amp/.style={mark options={solid},color=TUMBeamerYellow, line width=\lineWidth, mark size=\marksize, dashdotted}}
\tikzset{cae/.style={mark options={solid},color=TUMBeamerGreen, line width=\lineWidth, mark=triangle, mark size=\marksize, dashed}}
\tikzset{channelnet/.style={mark options={solid},color=black, line width=\lineWidth, mark=x, mark size=\marksize, dotted}}
\tikzset{cnn/.style={mark options={solid},color=black, line width=\lineWidth, mark=x, mark size=\marksize, dotted}}
\tikzset{genielmmse/.style={mark options={solid},color=TUMBeamerRed, line width=\lineWidth}}
\tikzset{globalcov/.style={mark options={solid},color=TUMBeamerOrange, line width=\lineWidth, mark=o, mark size=\marksize, dashed}}
\tikzset{gmm/.style={mark options={solid},color=TUMBeamerBlue, line width=\lineWidth, mark=square, mark size=\marksize, dashed}}
\tikzset{gmmdiag/.style={mark options={solid},color=mylila, line width=\lineWidth, mark=diamond, mark size=\marksize, dashdotted}}
\tikzset{ls/.style={mark options={solid},color=TUMMediumGray, line width=\lineWidth, mark=pentagon, mark size=\marksize, dashed}}
\tikzset{omp/.style={mark options={solid},color=TUMBeamerGreen, line width=\lineWidth, mark=triangle, mark size=\marksize, dashed}}

\tikzset{onePlow/.style={mark options={solid},color=TUMBeamerBlue, line width=\lineWidth, mark=square, mark size=\marksize, solid}}
\tikzset{onePhigh/.style={mark options={solid},color=TUMBeamerBlue, line width=\lineWidth, mark=square, mark size=\marksize, dashed}}
\tikzset{threePlow/.style={mark options={solid},color=black, line width=\lineWidth, mark=x, mark size=\marksize, solid}}
\tikzset{threePhigh/.style={mark options={solid},color=black, line width=\lineWidth, mark=x, mark size=\marksize, dashed}}
\tikzset{quaDlow/.style={mark options={solid},color=TUMBeamerRed, line width=\lineWidth, mark=diamond, mark size=\marksize, solid}}
\tikzset{quaDhigh/.style={mark options={solid},color=TUMBeamerRed, line width=\lineWidth, mark=diamond, mark size=\marksize, dashed}}

\begin{document}

\title{An Asymptotically MSE-Optimal Estimator\\based on Gaussian Mixture Models}

\author{Michael~Koller, Benedikt~Fesl, Nurettin~Turan, and
        Wolfgang~Utschick,~\IEEEmembership{Fellow,~IEEE}
\thanks{This work was supported by the Deutsche Forschungsgemeinschaft under
grant UT 36/21.}%
\thanks{Preliminary results have been published at ICASSP'22~\cite{KoFeTuUt22}.}%
\thanks{The authors are with Professur für Methoden der Signalverarbeitung, Technische Universität München,
80333 München, Germany, e-mail: \{michael.koller,benedikt.fesl,nurettin.turan,utschick\}@tum.de}}

\markboth{}
{Koller \MakeLowercase{\textit{et al.}}}


\maketitle

\begin{abstract}
This paper investigates a channel estimator based on \acp{gmm} in the context of linear inverse problems with additive Gaussian noise.
We fit a \ac{gmm} to given channel samples to obtain an analytic \ac{pdf} which approximates the true channel \ac{pdf}.
Then, a \ac{cme} corresponding to this approximating \ac{pdf} is computed in closed form and used as an approximation of the optimal \ac{cme} based on the true channel \ac{pdf}.
This optimal \ac{cme} cannot be calculated analytically because the true channel \ac{pdf} is generally unknown.
We present mild conditions which allow us to prove the convergence of the \ac{gmm}-based \ac{cme} to the optimal \ac{cme} as the number of \ac{gmm} components is increased.
Additionally, we investigate the estimator's computational complexity and present simplifications based on common model-based insights.
Further, we study the estimator's behavior in numerical experiments including \ac{mimo} and wideband systems.
\end{abstract}


%
\IEEEpeerreviewmaketitle

\glsresetall

\section{Introduction}

Channel estimation plays a critical role in future mobile communications systems, e.g., \cite{RuPeLaLaMaEdTu13,ArDeAxMo14,SaRa16}.
The \ac{mse} minimizing channel estimator is known as \ac{cme}.
Computing the \ac{cme} in closed form requires analytic knowledge of the channel \ac{pdf}.
Even if the \ac{pdf} was known, calculating the \ac{cme} might not be possible analytically or not be tractable practically.
Increasingly, advanced channel models (e.g., \cite{3gpp}) or simulators (e.g., \cite{QuaDRiGa1, QuaDRiGa2,Al19}) are used to generate large amounts of realistic channel samples.
In a real application, channel samples can, for example, be collected at the base station to be used in addition to or instead of the simulated data.
Importantly, such data represent the whole scenario (or environment) in which the base station is placed.
It is thus interesting to investigate data-based algorithms to design channel estimators which are applicable to a whole scenario.
Many estimators have been proposed in this context
(cf., e.g., \Cref{sec:sota} for details).
In particular, estimators based on \ac{cs} and on machine learning have recently been proposed (see, e.g., \cite{BuHuMuDa18,HaMaZwDa20}).
To our knowledge, the algorithms' optimality has not been studied if an arbitrary channel \ac{pdf} is assumed.

In this paper, we study a \ac{gmm}-based channel estimator.
The estimator itself has already been investigated in the case where the channel \ac{pdf} is given by a \ac{gmm}.
One of our contributions is to provide a strong motivation to employ the \ac{gmm}-based estimator even if the channel \ac{pdf} is not a \ac{gmm}.
To this end, we show (in a proof and in numerical simulations) that even if the channel is not \ac{gmm} distributed,
the \ac{gmm}-based estimator converges to the optimal \ac{cme} as the number of \ac{gmm} components increases.

In detail, the following approach is taken in this paper.
First, channel samples are used to fit a \ac{gmm}.
Since \acp{gmm} can approximate any continuous \ac{pdf}~\cite{NgNgChMc20},
the fitted \ac{gmm} is a \ac{pdf} which approximates the unknown true channel \ac{pdf}.
Second, we analytically compute a \ac{cme} for channels distributed according to the \ac{gmm} \ac{pdf}.
Since the \ac{gmm} \ac{pdf} approximates the true channel \ac{pdf},
we ask whether the \ac{gmm}-based \ac{cme} approximates the true \ac{cme}.

A related work is~\cite{GuZh19}, where the authors assume that the channel is \ac{gmm}-distributed and study the available closed-form \ac{cme} for example in the asymptotic high \ac{snr} regime to derive pilot signals.
In this work, we do not assume that the channel is \ac{gmm}-distributed.
Instead, a \ac{gmm} is used as an approximation of the true channel \ac{pdf} and we analyze whether the corresponding \ac{cme} is an approximation of the true \ac{cme}.
This can be viewed as a study of the \ac{gmm} estimator in the high number of \ac{gmm} components regime.
A main contribution of our paper is to prove that as the number of components increases, the \ac{gmm} \ac{cme} converges to the optimal \ac{cme} if the observation matrix is invertible (cf.~\Cref{thm:main_result}).
For noninvertible observation matrices, we make a weaker statement.
Moreover, we analyze the \ac{gmm} estimator's computational complexity and show how the complexity can be reduced in different estimation scenarios.

We study the \ac{gmm} estimator in numerical simulations where we consider both \ac{mimo} and wideband channel estimation scenarios with both invertible and noninvertible observation matrices.
The considered channel data come from a 3GPP channel model \cite{3gpp} and from the QuaDRiGa channel simulator~\cite{QuaDRiGa1, QuaDRiGa2} so that they are not \ac{gmm}-distributed by construction.
The generated data represent a scenario where for example a base station covers a certain sector with users whose positions are drawn uniformly at random.
The obtained \ac{gmm} estimator is then suited for channel estimation in the whole scenario.
Already for a finite number of \ac{gmm} components,
the \ac{gmm} estimator shows a performance close to the optimal \ac{cme} in the numerical simulations.
We emphasize the \ac{gmm} estimator's broad applicability by comparing it to state-of-the-art algorithms from the literature.

The paper is structured as follows.
\Cref{sec:signal_model} introduces the signal model discussed throughout the paper as well as particular instances thereof which are used in numerical simulations.
\Cref{sec:gmm_literature} reviews \acp{gmm} and channel estimation literature which employs them.
The main part is \Cref{sec:main_part} where we investigate the \ac{gmm}-based \ac{cme} and study its convergence to the optimal \ac{cme} as well as its computational complexity.
\Cref{sec:sota,sec:numerical_results} present state-of-the-art channel estimation algorithms, channel models, and numerical simulations.

\emph{Notation:}
The supremum norm of a continuous function \( f: \R^N \to \R \) is given by \( \| f \|_\infty = \sup_{\mbx\in\R^N} |f(\mbx)| \),
and \( \| \mbx \| \) is the Euclidean norm of \( \mbx \in \C^N \).
A real- or complex-valued normal distribution with mean vector \( \mbmu \) and covariance matrix \( \mbC \) is denoted by \( \calN(\mbmu, \mbC) \) or \( \calN_{\C}(\mbmu, \mbC) \), respectively.
The vectorization (stacking columns) of a matrix \( \mbX \in \C^{m\times N} \) is written as \( \vect(\mbX) \in \C^{mN} \), and \( \mbA \otimes \mbB \in \C^{m_1m_2\times N_1N_2} \) is the Kronecker product of \( \mbA \in \C^{m_1\times N_1} \) and \( \mbB \in \C^{m_2\times N_2} \).

\section{Signal Models}\label{sec:signal_model}

We consider the generic signal model
\begin{equation}\label{eq:signal_model}
    \mby
    = \mbA \mbh + \mbn, \quad\quad \mbn \sim \calN_{\C}(\mbzero, \mbSigma)
\end{equation}
where \( \mbh \in \C^N \) is the channel, \( \mbA \in \C^{m\times N} \) is the observation matrix, and \( \mbn \in \C^m \) is additive white Gaussian noise.
The technical interpretation (e.g., number of antennas or pilots) of the dimensions \( m \) and \( N \) depends on the context.
Examples can be found in the following subsections.
The observation \( \mby \in \C^m \), the matrix \( \mbA \), the noise mean vector \( \mbzero \in \C^m \), and the noise covariance matrix \( \mbSigma \in \C^{m\times m} \) are given.
The goal of channel estimation is to recover \( \mbh \) from~\eqref{eq:signal_model}.

In this paper, we study a channel estimation algorithm which is designed using the given signal model as well as a data set of channel samples.
While the main part addresses the generic signal model~\eqref{eq:signal_model}, we consider the following three instances of it in numerical simulations.

\subsection{Single-Input Multiple-Output Signal Model}\label{sec:signal_model_simo}

The \ac{simo} signal model is for instance appropriate if a single-antenna mobile device transmits pilot signals to a base station with \( N \) antennas which receives
\begin{equation}\label{eq:signal_model_simo}
    \mby = \mbh + \mbn \in \C^N.
\end{equation}
This model is interesting for us because the observation matrix is the identity matrix and therefore invertible.
Further, the performance of the proposed channel estimator can be studied without having to take into account the difficulty of choosing a suitable observation matrix.

\subsection{Multiple-Input Multiple-Output Signal Model}\label{sec:signal_model_mimo}

If a mobile user with \( \Ntx \) antennas transmits \( N_p \) pilots to a base station with \( \Nrx \) antennas,
the receive signal \( \mbY \in \C^{\Nrx\times N_p} \) can be written as
\begin{equation}\label{eq:signal_model_mimo}
    \mbY = \mbH \mbP + \mbN
\end{equation}
where \( \mbH \in \C^{\Nrx\times \Ntx} \) is the channel, \( \mbP \in \C^{\Ntx\times N_p} \) is the pilot, and \( \mbN \in \C^{\Nrx \times N_p} \) is the noise matrix.
With the definitions \( \mbh = \vect(\mbH) \), \( \mby = \vect(\mbY) \), \( \mbn = \vect(\mbN) \), and \( \mbA = \mbP^\tp \otimes \mbI_{\Nrx} \), the \ac{mimo} signal model~\eqref{eq:signal_model_mimo} is an instance of~\eqref{eq:signal_model}.

\subsection{Wideband Signal Model}\label{sec:signal_model_wideband}

If we consider a \ac{siso} transmission in the spatial domain over a frequency-selective fading channel, \( \mbH\in \C^{N_c\times N_t} \) represents the time-frequency response of the channel for \( N_c \) subcarriers and \( N_t \) time slots.
When only \( N_p \) positions of the time-frequency response are occupied with pilot symbols,
then there is a \textit{selection matrix} \( \mbA \in \{0,1\}^{N_p\times N_c N_t} \) which represents the pilot positions.
This leads to the observations as described in~\eqref{eq:signal_model} with \( \mbh =\vect(\mbH)\in\C^{N_c N_t} \). 
Regarding the structure of the pilot positions, three different arrangements are commonly considered: block-, comb-, and lattice-type, cf. \cite{CoErPuBa02}.

\section{Gaussian Mixture Models in the Literature}\label{sec:gmm_literature}

In this section, we briefly explain \acp{gmm} and summarize channel estimation literature which makes use of \acp{gmm}.

\subsection{Gaussian Mixture Models}\label{sec:gmms}

A \ac{gmm} with \( K \) components is a \ac{pdf} of the form~\cite{bookBi06}
\begin{equation}\label{eq:gmm_of_h}
    f_{\mbh}^{(K)}(\mbh) = \sum_{k=1}^K p(k) \calN_{\C}(\mbh; \mbmu_k, \mbC_k)
\end{equation}
consisting of a weighted sum of \( K \) Gaussian \acp{pdf}.
The probabilities \( p(k) \) are called \textit{mixing coefficients}, and \( \mbmu_k \in \C^N \) and \( \mbC_k \in \C^{N\times N} \) denote the mean vector and covariance matrix of the \( k \)th \ac{gmm} component, respectively.
As explained in~\cite{bookBi06}, \acp{gmm} allow to calculate the \textit{responsibilities} \( p(k \mid \mbh) \) by evaluating Gaussian likelihoods:
\begin{equation}
    p(k \mid \mbh) = \frac{p(k) \calN_{\C}(\mbh; \mbmu_k, \mbC_k)}{\sum_{i=1}^K p(i) \calN_{\C}(\mbh; \mbmu_i, \mbC_i) }.
\end{equation}
This is an important property for our considerations.

Given data samples, an \ac{em} algorithm can be used to fit a \( K \)-components \ac{gmm}~\cite{bookBi06}.
The data-based fitting process determines the mixing coefficients, the mean vectors, and the covariance matrices.
A detailed introduction to \acp{gmm} and the corresponding well-known \ac{em} algorithm can, e.g., be found in~\cite{bookBi06}.

\subsection{Gaussian Mixture Models in Channel Estimation Literature}

In \cite{ViSc13}, channels are modeled as sparse vectors whose non-zero coefficients are \ac{gmm}-distributed.
A combination of \ac{em} and \ac{amp} is then introduced for channel estimation.
The algorithm simultaneously estimates the \ac{gmm} parameters.
Building on this work, \cite{WeJiWoChTi15} models the beam domain channels via \acp{gmm} in the context of uplink channel estimation with pilot contamination.
The authors of~\cite{SuWa19} then investigate the approach further including a new initialization technique for the algorithm.
In~\cite{WeHuDa21}, the beam domain channel is also assumed to be \ac{gmm}-distributed and a modification of learned \ac{amp} is proposed for sparse channel estimation.

The authors of~\cite{MuMiDe20} model temporal channel variations as \acp{gmm}, e.g., in order to predict channel states.
The authors of~\cite{ZhXuMeZhXu19} propose to improve channel estimation techniques by using \acp{gmm} as a better characterization of the noise in communications environments than it is given by the additive white Gaussian noise model.
In~\cite{NaRa18}, a \ac{gmm} prior is used for the unknown data symbols in semi-blind channel estimation.

\acp{gmm} are also employed for channel clustering tasks.
For example, \cite{LiZhMaZh20} use \acp{gmm} for channel multipath clustering.
In~\cite{LiZhTaTi19}, a \textit{power weighted \ac{gmm}} is proposed to increase the clustering performance.
Another variation of \acp{gmm}, called \textit{rotationally invariant \ac{gmm}}, can be found in~\cite{JiYuWaXuYu20}.

In \cite{GuZh18,GuZh19}, the true channel \ac{pdf} is assumed to be equal to a \ac{gmm} and the authors then investigate the corresponding \ac{cme} to optimize the pilot matrix.
To this end, the asymptotic high-\ac{snr} regime of the \ac{cme} is studied.
Further, an information-theoretic criterion for pilot optimization is introduced because the \ac{mse} of the estimator has no closed-form expression and is thus not suitable as optimization criterion~\cite{GuZh19}.

In this paper, we study the same \ac{gmm}-based estimator as the authors of \cite{GuZh18,GuZh19}.
However, we do not assume that the true channel \ac{pdf} is equal to a \ac{gmm}.
Instead, we take the \ac{gmm} as an approximation of the true channel \ac{pdf} and we ask whether the corresponding \ac{gmm}-based estimator is then an approximation of the true \ac{cme} (based on the true channel \ac{pdf}).
In this sense, we study the \ac{gmm}-based estimator's behavior in the high number of components regime and our work complements~\cite{GuZh19} by motivating the application of the estimator in a wider class of channel models.

\section{Main Part}\label{sec:main_part}

The \ac{mse}-optimal channel estimate for the model~\eqref{eq:signal_model} is given by the conditional expectation \( \expec[\mbh \mid \mby] \), cf., e.g.,~\cite{bookSc16}.
However, the true channel \ac{pdf} is generally not known and, therefore, \( \expec[\mbh \mid \mby] \) can generally not be calculated analytically.
Even if the true channel \ac{pdf} was known, the \ac{cme} \( \expec[\mbh \mid \mby] \) might still not have an analytic expression.
In this section, we investigate a \ac{gmm}-based \ac{cme} with closed-form expression and prove that it converges to the optimal \ac{cme} as the number of \ac{gmm} components is increased.
Further, we discuss its computational complexity and how the complexity can be reduced.

\subsection{Channel Estimator}\label{sec:gmm_channel_estimator}

\acp{gmm} are known to be able to approximate any continuous \ac{pdf} arbitrarily well~\cite{NgNgChMc20}.
In particular, if \( f_{\mbh} \) denotes the \ac{pdf} of the channel which is assumed to be continuous, then there exists a sequence \( ( f_{\mbh}^{(K)} )_{K=1}^\infty \) of \acp{gmm} which converges uniformly to \( f_{\mbh} \).
To define a \ac{gmm}-based estimator, let \( f_{\mbn} \) and \( f_{\mby} \) be the \acp{pdf} of the noise and the observation, respectively, and let us first observe the following:
\begin{equation}\label{eq:fh_given_y}
    f_{\mbh\mid\mby}(\mbh\mid\mby) = \frac{f_{\mby\mid\mbh}(\mby\mid\mbh)f_{\mbh}(\mbh)}{f_{\mby}(\mby)} = \frac{f_{\mbn}(\mby - \mbA\mbh)f_{\mbh}(\mbh)}{f_{\mby}(\mby)}.
\end{equation}
With this the optimal \ac{cme} can be expressed as
\begin{equation}\label{eq:conditional_mean}
    \hhat(\mby) = \expec[\mbh \mid \mby]
    = \int \mbh \frac{f_{\mbn}(\mby - \mbA\mbh)f_{\mbh}(\mbh)}{f_{\mby}(\mby)} d \mbh.
\end{equation}
For every \( K \in \N \), we now consider the model
\begin{equation}
    \mby^{(K)} = \mbA \mbh^{(K)} + \mbn
\end{equation}
where \( \mbh^{(K)} \) is distributed according to the \ac{gmm} \( f_{\mbh}^{(K)} \)
which has the form \eqref{eq:gmm_of_h}.
Since we have a sequence \( (f_{\mbh}^{(K)} )_{K=1}^\infty \) of \acp{gmm}, the parameters \(p(k)\), \(\mbmu_k\), and \(\mbC_k\) would also depend on the sequence index \(K\) but we omit it for readability.
Let \( f_{\mby}^{(K)} \) be the \ac{pdf} of \( \mby^{(K)} \).
We now define a \ac{gmm}-based estimator
\begin{equation}\label{eq:conditional_mean_K}
    \hhat^{(K)}(\mby) := \expec^{(K)}[\mbh^{(K)} \mid \mby]
    :=\int \mbh \frac{f_{\mbn}(\mby - \mbA\mbh)f_{\mbh}^{(K)}(\mbh)}{f_{\mby}^{(K)}(\mby)} d \mbh
\end{equation}
by replacing \( f_{\mbh} \) and \( f_{\mby} \) in~\eqref{eq:conditional_mean} with \( f_{\mbh}^{(K)} \) and \( f_{\mby}^{(K)} \), respectively,
because similar to~\eqref{eq:fh_given_y} we have
\begin{equation}\label{eq:fhk_given_y}
    f_{\mbh\mid\mby}^{(K)}(\mbh\mid\mby)
    = \frac{f_{\mbn}(\mby - \mbA\mbh)f_{\mbh}^{(K)}(\mbh)}{f_{\mby}^{(K)}(\mby)}.
\end{equation}

The law of total expectation allows us to write
\begin{equation}\label{eq:total_expectation}
    \hhat^{(K)}(\mby) = \sum_{k=1}^K p(k\mid \mby) \expec^{(K)}[\mbh^{(K)} \mid \mby, k]
\end{equation}
in order to introduce the \ac{gmm} mixing variable.
By definition of \acp{gmm}, conditioning on one of the mixing variables yields a Gaussian random vector.
That is, \( \mbh^{(K)} \mid k \sim \calN_{\C}(\mbmu_k, \mbC_k) \) is the \( k \)th Gaussian in the \ac{gmm} \( f_{\mbh}^{(K)} \), see also~\eqref{eq:gmm_of_h}.
Since \( \mbh^{(K)} \mid k \) is Gaussian, also \( \mbA \mbh^{(K)} \mid k \) and therefore \( \mby^{(K)} \mid k \) are Gaussian.
The conditional mean vector and conditional covariance matrix of \( \mby^{(K)} \mid k \) are \( \expec[\mby^{(K)} \mid k] = \mbA \meanhk \) and
\begin{equation}
    \expec[(\mby^{(K)} - \mbA \meanhk)(\mby^{(K)} - \mbA \meanhk)^\herm \mid k]
    = \mbA \covhk \mbA^\herm + \mbSigma,
\end{equation}
respectively.
The well-known \ac{lmmse} formula can now be used to compute
\begin{equation}\label{eq:lmmse_formula}
    \expec^{(K)}[\mbh^{(K)} \mid \mby, k] =
 \covhk \mbA^\herm (\mbA \covhk \mbA^\herm + \mbSigma)^{-1} (\mby - \mbA \meanhk) + \meanhk
\end{equation}
which can be plugged into~\eqref{eq:total_expectation}.

In order to calculate \( p(k \mid \mby) \) in~\eqref{eq:total_expectation}, we compute the \ac{pdf}
\begin{equation}\label{eq:gmm_y}
    f_{\mby}^{(K)}(\mby) = \sum_{k=1}^K p(k) \calN_{\C}(\mby; \mbA \meanhk, \mbA \covhk \mbA^\herm + \mbSigma),
\end{equation}
which is a \ac{gmm}.
\acp{gmm} allow to calculate the responsibilities by evaluating Gaussian likelihoods (cf. \Cref{sec:gmms}):
\begin{equation}\label{eq:responsibilities}
    p(k \mid \mby) = \frac{p(k) \calN_{\C}(\mby; \mbA \meanhk, \mbA \covhk \mbA^\herm + \mbSigma)}{\sum_{i=1}^K p(i) \calN_{\C}(\mby; \mbA \meanhi, \mbA \covhi \mbA^\herm + \mbSigma) }.
\end{equation}
Plugging this into~\eqref{eq:total_expectation} shows that as soon as the mixing coefficients \( p(k) \) as well as the means \( \meanhk \) and covariances \( \covhk \) are given, the estimator \( \hest \) can be computed in closed form by combining~\eqref{eq:total_expectation}, \eqref{eq:lmmse_formula}, and~\eqref{eq:responsibilities}, which results in~\eqref{eq:estimator_full}.
As discussed next in \Cref{sec:computational_complexity}, various quantities of the \ac{gmm} estimator can be precomputed at this point to save computational complexity:
the products involving the known observation matrix \( \mbA \) and means \( \mbmu_k \) and covariance matrices \(\mbC_k\) as well as the \ac{lmmse} filters including the computationally costly matrix inverse in \eqref{eq:estimator_full}.
To obtain the \ac{gmm} parameters, the channel \ac{pdf} \( f_{\mbh} \) needs to be approximated by fitting a \( K \)-components \ac{gmm} to given channel samples, cf.~\Cref{sec:gmms}.

\newcounter{MYtempeqncnt}
\begin{figure*}[!t]
\normalsize
\setcounter{MYtempeqncnt}{\value{equation}}
\setcounter{equation}{15}
\begin{equation}\label{eq:estimator_full}
    \hhat^{(K)}(\mby) = \sum_{k=1}^K
    \frac{p(k) \calN_{\C}(\mby; \mbA \meanhk, \mbA \covhk \mbA^\herm + \mbSigma)}{\sum_{i=1}^K p(i) \calN_{\C}(\mby; \mbA \meanhi, \mbA \covhi \mbA^\herm + \mbSigma) }
    \left(\covhk \mbA^\herm (\mbA \covhk \mbA^\herm + \mbSigma)^{-1} (\mby - \mbA \meanhk) + \meanhk\right)
\end{equation}
\setcounter{MYtempeqncnt}{\value{equation}}
\setcounter{equation}{\value{MYtempeqncnt}}
\hrulefill
\vspace*{4pt}
\end{figure*}

A possible application scenario of the \ac{gmm} estimator would be to use channel samples collected at, for example, the base station of a cellular radio system to construct a site-specific \ac{gmm} channel estimator.
In an initial (offline) training phase, the channel samples are used to fit a \( K \)-components \ac{gmm}.
Afterwards, (online) channel estimates are computed via~\eqref{eq:estimator_full}.
The formula~\eqref{eq:estimator_full} can also be found, e.g., in~\cite{GuZh19}.
One of the key differences to other work is that we provide a strong motivation to use~\eqref{eq:estimator_full} even if the channel \ac{pdf} \( f_{\mbh} \) is not a \ac{gmm}.
\Cref{alg:gmm_estimator} summarizes both the offline \ac{gmm} training and the online channel estimation phases.
The necessary number \( M \) of training data depends, e.g., on the number \( K \) of \ac{gmm} components.
We discuss this in more detail in \Cref{sec:numerical_results}.

\begin{algorithm}[!t]
    \caption{GMM Estimator}
    \label{alg:gmm_estimator}
    \begin{algorithmic}[1]
        \renewcommand{\algorithmicensure}{\textbf{Offline GMM Training Phase}}
        \ENSURE
        \REQUIRE training data \( \{ \mbh_m \}_{m=1}^M \), number of components \( K \)
        \STATE \( (\{ p(k) \}_{k=1}^K, \{ \mbmu_k \}_{k=1}^K, \{ \mbC_k \}_{k=1}^K) \leftarrow \mathrm{EM}(\{ \mbh_m \}_{m=1}^M, K) \)
        \COMMENT{an EM-algorithm computes all parameters of \( f_{\mbh}^{(K)} \)}
        \par\vskip.5\baselineskip\hrule height .4pt\par\vskip.5\baselineskip
        \renewcommand{\algorithmicensure}{\textbf{Online Channel Estimation Phase}}
        \ENSURE
        \REQUIRE observation \( \mby \), matrix \( \mbA \), noise matrix \( \mbSigma \)
        \STATE \( \hhat^{(K)} \leftarrow \mbzero \)
        \renewcommand{\algorithmicendfor}{\textbf{end}}
        \FOR {\( k = 1 \) to \( K \)}
        \STATE \( p(k \mid \mby) \leftarrow \frac{p(k) \calN_{\C}(\mby; \mbA \meanhk, \mbA \covhk \mbA^\herm + \mbSigma)}{\sum_{i=1}^K p(i) \calN_{\C}(\mby; \mbA \meanhi, \mbA \covhi \mbA^\herm + \mbSigma) } \)
        \STATE \( \tilde{\mbh} \leftarrow \covhk \mbA^\herm (\mbA \covhk \mbA^\herm + \mbSigma)^{-1} (\mby - \mbA \meanhk) + \meanhk \)
        \STATE \( \hhat^{(K)} \leftarrow \hhat^{(K)} + p(k \mid \mby) \tilde{\mbh} \)
        \ENDFOR\,
        \COMMENT{many quantities in the loop can be precomputed}
        \RETURN \( \hhat^{(K)} \)
        \COMMENT{estimated channel, see~\eqref{eq:estimator_full}}
    \end{algorithmic}
\end{algorithm}

\subsection{Computational Complexity}\label{sec:computational_complexity}

To compute \( \hest(\mby) \) in~\eqref{eq:estimator_full}, \( K \) responsibilities \( p(k \mid \mby) \)~\eqref{eq:responsibilities} and \( K \) \ac{lmmse} formulas~\eqref{eq:lmmse_formula} need to be evaluated.
Since both the matrix \( \mbA \) and the \ac{gmm} covariance matrices \( \covhk \) do not change between observations,
the inverse in~\eqref{eq:lmmse_formula} can be precomputed offline for various \acp{snr}.
Thus, evaluating~\eqref{eq:lmmse_formula} online is dominated by matrix-vector multiplications and has a complexity of \( \calO(mN) \).
The responsibilities are calculated by evaluating Gaussian densities, as can be seen from~\eqref{eq:responsibilities}.
A Gaussian density with mean \( \mbmu \in \C^{m} \) and covariance matrix \( \mbC \in \C^{m\times m} \) can be written as
\begin{equation}\label{eq:gaussian_density}
    \calN_{\C}(\mbx; \mbmu, \mbC) = \frac{\exp(-(\mbx - \mbmu)^\herm \mbC^{-1} (\mbx - \mbmu))}{\pi^{m} \det(\mbC)}.
\end{equation}
Again, since the \ac{gmm} covariance matrices and mean vectors do not change between observations, the inverse and the determinant of the densities in~\eqref{eq:responsibilities} can be precomputed offline.
Thus, the online evaluation is again dominated by matrix-vector multiplications and has a complexity of \( \calO(m^2) \).
The resulting overall complexity of computing \( \hest \) is \( \calO(K m N) \).

In some cases, as demonstrated in the following subsections, the computational complexity can be reduced by constraining the \ac{gmm} covariance matrices \( \covhk \) such that corresponding matrix-vector multiplications are accelerated.
In other cases, the number of \ac{gmm} parameters might be reduced by introducing covariance matrix constraints which can enhance the convergence of the \ac{em} algorithm, improve the resulting estimation performance, and reduce the required amount of channel samples.
Particular choices for constraints can come from scenario-specific insights.
We demonstrate the feasibility of the following two constraint examples in \Cref{sec:numerical_results}.

\subsubsection{Circulant covariance matrices}\label{sec:circulant_covs}

A first example is a scenario, where the base station employs a \ac{ula} and where the channel covariance matrix therefore is Toeplitz structured.
For large numbers of antennas, a Toeplitz matrix is well approximated by a circulant matrix~\cite{Gr06}.
Any circulant matrix \( \mbC \in \C^{N\times N} \) has an eigendecomposition of the form \( \mbC = \mbF^\herm \diag(\mbc) \mbF \)
where \( \mbF \in \C^{N\times N} \) is the \ac{dft} matrix and where \( \mbc \in \C^N \).
Consequently, thanks to fast Fourier transforms, matrix-vector multiplications involving circulant matrices can be performed in \( \calO(N \log(N)) \) time.
For a large number of antennas, we therefore have a motivation to use circulant covariance matrices \( \covhk = \mbF^\herm \diag(\mbc_k) \mbF \) in the \ac{gmm}.

This is particularly interesting for a signal model where \( \mbA = \mbI \) and \( \mbSigma = \sigma^2 \mbI = \sigma^2 \mbF \mbF^\herm \).
In this case, the \ac{lmmse} formula~\eqref{eq:lmmse_formula} simplifies to
\begin{equation}\label{eq:nlogn_lmmse_formula}
    \expec[\mbh \mid \mby, k] = \mbF^\herm \diag(\mbd_k) \mbF (\mby - \meanhk) + \meanhk
\end{equation}
where the \( i \)th entry of the vector \( \mbd_k \) is given by \( [\mbd_k]_i = \frac{[\mbc_k]_i}{[\mbc_k]_i + \sigma^2} \),
such that~\eqref{eq:nlogn_lmmse_formula} can be calculated in \( \calO(N \log(N)) \) time.
With a circulant \( \mbC = \mbF^\herm \diag(\mbc) \mbF \), \eqref{eq:gaussian_density} reads as
\begin{multline}\label{eq:nlogn_gaussian_density}
    \calN_{\C}(\mbx; \mbmu, \mbF^\herm \diag(\mbc) \mbF) = \\ \frac{ \exp(-(\mbF(\mbx - \mbmu))^\herm \diag(\mbc)^{-1} \mbF (\mbx - \mbmu))}{\pi^n \prod_{i=1}^N [\mbc]_i}.
\end{multline}
A first observation is that this can also be evaluated in \( \calO(N \log(N)) \) time such that computing channel estimates \( \hest \) has a complexity of \( \calO(K N \log(N)) \) if circulant covariance matrices are used in the \ac{gmm}.
A second observation is that the Gaussian density~\eqref{eq:nlogn_gaussian_density} has a significantly reduced number of parameters: \( N + N \) (mean vector \( \mbmu \) and covariance vector \( \mbc \)) in contrast to \( N + \frac{N(N+1)}{2} \) (mean vector \( \mbmu \) and covariance matrix \( \mbC \)) in~\eqref{eq:gaussian_density}.
This simplifies the \ac{em} algorithm iterations of the \ac{gmm} fitting process and reduces the number of required training channel samples.
For this latter reason, even if \( \mbA \neq \mbI \), one might be interested in employing a \ac{gmm} with circulant covariance matrices.
The relationship between the number of channel samples and the \ac{em} algorithm's performance is demonstrated in \Cref{sec:numerical_results}.
In an implementation, instead of constraining the covariance matrices to be circulant,
all channel samples can be Fourier transformed as \( \tilde{\mbh} = \mbF \mbh \) and then the \ac{gmm}'s covariance matrices can be constrained to be diagonal matrices due to the relation \( \mbC = \mbF^\herm \diag(\mbc) \mbF \).

\subsubsection{Kronecker covariance matrices}\label{sec:kronecker_covs}

Another example where complexity can be reduced is the \ac{mimo} signal model from \Cref{sec:signal_model_mimo}.
A well-known assumption for spatial correlation scenarios is that the scattering in the vicinity of the transmitter and of the receiver are independent of each other, cf. \cite{KeScPeMoFr02}.
In this case, every channel covariance matrix \( \mbC \) can be decomposed into the Kronecker product of a transmit and receive side spatial covariance matrix: \( \mbC = \mbC_{\text{tx}} \otimes \mbC_{\text{rx}} \).
Here, we have a motivation to use a \ac{gmm} with Kronecker product covariance matrices \( \covhk = \mbC_{\text{tx},k} \otimes \mbC_{\text{rx},k} \).

To this end, instead of fitting a single \ac{gmm} using the vectorized channel data of dimension \( N = \Ntx \Nrx \),
one can fit two independent transmit and receive side \acp{gmm} of dimensions \( \Ntx \) and \( \Nrx \), respectively.
This not only results in lower offline training complexity and in the ability to parallelize, but also in a smaller number of training channel samples needed because the respective \acp{gmm} have much fewer parameters.
The training channel samples for these low-dimensional \acp{gmm} are obtained by taking the rows (columns) of the available channel matrices independently in order to fit the transmit (receive) side \ac{gmm}.
In order to then obtain the full-size covariance matrices \( \mbC_k \),
all combinatorial Kronecker products of transmit and receive side covariance matrices \( \mbC_{\text{tx},i} \) and \( \mbC_{\text{rx},j} \) are computed.
The details are described in the numerical simulations section.

In this example, plugging the Kronecker decomposition \( \covhk = \mbC_{\text{tx},k} \otimes \mbC_{\text{rx},k} \) into the \ac{lmmse} formula~\eqref{eq:lmmse_formula} does not lead to an expression that simplifies to a Kronecker product.
This is because the inverse in~\eqref{eq:lmmse_formula} can generally not be written in terms of a Kronecker product and, thus, full matrix-vector products are necessary.
However, \cite{SiMeWrRu10} explains how~\eqref{eq:lmmse_formula} can be approximated by means of a Kronecker product in the described setting, which might be interesting if computational complexity of~\eqref{eq:lmmse_formula} is an issue.
Nonetheless, even if the online computational complexity is not affected, Kronecker \ac{gmm} covariance matrices can still be beneficial, for instance, if the number of available channel samples is small.
We demonstrate this case in the numerical simulations section.

\subsection{Convergence of the Estimator}\label{sec:convergence_theorem}

This subsection uses a universal approximation result of~\cite{NgNgChMc20} to show that if \( f_{\mbh} \) is continuous, then the \ac{gmm}-based estimator \( \hest \) in~\eqref{eq:conditional_mean_K} can approximate the optimal \ac{cme} \( \hhat(\mby) \) in~\eqref{eq:conditional_mean} arbitrarily well as the number \( K \) of \ac{gmm} components increases.
The intuition is that if a sequence of \acp{pdf} \( f_{\mbh}^{(K)} \), which is used in~\eqref{eq:conditional_mean_K}, converges to the channel \ac{pdf} \( f_{\mbh} \),
we can conjecture that also \( f_{\mby}^{(K)} \) from~\eqref{eq:gmm_y} converges to \( f_{\mby} \) and that then \( \hhat^{(K)}(\mby) \) in~\eqref{eq:conditional_mean_K} converges to the \ac{cme} \( \hhat(\mby) \) in~\eqref{eq:conditional_mean}.

Recall that the \ac{pdf} of a complex random vector can be expressed by means of a joint \ac{pdf} of its real and imaginary parts.
Therefore, this subsection considers real-valued quantities only and the results generalize to the complex-valued setting by considering stacked real and imaginary parts.

To state the main result, we adopt some definitions from~\cite{NgNgChMc20}.
Let \( \calC = \{ f: \R^N \to \R: f \geq 0, \int f(\mbx) d\mbx = 1, f \text{ is contin.} \} \) denote the set of all continuous \acp{pdf}.
Further, let \( g \) denote the standard Gaussian density and define the class of \( K \)-component location-scale finite Gaussian mixtures as
\begin{equation}
    \calM_K = \left\{ h: h(\mbx) = \sum_{k=1}^K c_k \frac{1}{\sigma_k^N} g\left( \frac{\mbx - \mbmu_k}{\sigma_k} \right) \right\}
\end{equation}
with \( \mbmu_k \in \R^N, \sigma_k > 0, c_k \geq 0 \) for all \( k \in \{1, \dots, K\} \) and \( \sum_{k=1}^K c_k = 1 \).
Then, any continuous \ac{pdf} can be approximated arbitrarily well by means of \acp{gmm}, as \cite[Theorem 5]{NgNgChMc20} states:
\begin{theorem}\label{thm:nguyens}
    Let
    \(
        \calC_0 = \{ f \in \calC: \forall \varepsilon > 0, \exists \text{ a compact } \calK \subset \R^N \text{ such that } \sup_{\mbx \in \R^N \setminus \calK} | f(\mbx) | < \varepsilon \}
    \)
    denote the set of all continuous \acp{pdf} which vanish at infinity.
    For any \( f \in \calC_0 \), there exists a sequence \( ( f^{(K)} )_{K=1}^\infty \) with \( f^{(K)} \in \calM_K \) with
    \begin{equation}
        \lim_{K\to\infty} \|f - f^{(K)}\|_\infty = 0.
    \end{equation}
\end{theorem}
Note that since a \ac{pdf} is integrable, it always vanishes at infinity such that this is not a constraint for our considerations.
As mentioned, we now work with real quantities \( \mby = \mbA \mbh + \mbn \) where \( \mbn \) is a real Gaussian random vector with mean zero and covariance matrix \( \mbSigma \in \R^{N\times N} \) whose \ac{pdf} we denote by \( f_{\mbn} \in \calC_0 \).
The \ac{pdf} of \( \mbh \) is \( f_{\mbh} \) and the \ac{pdf} of \( \mby \) is \( f_{\mby} \).
The following theorem is proved in Appendix~\ref{sec:proof_of_main_result}.

\begin{theorem}\label{thm:main_result}
    With the notation defined above, let \( \mbA \in \R^{N\times N} \) be invertible and let \( f_{\mbh} \in \calC_0 \) be arbitrary.
    Let \( ( f_{\mbh}^{(K)} )_{K=1}^\infty \) be a sequence of \acp{pdf} in \( \calC_0 \) which converges uniformly to \( f_{\mbh} \).
    Then, the estimator
    \begin{equation}
        \hhat^{(K)}(\mby) = \expec^{(K)}[\mbh^{(K)}\mid\mby] = \int \mbh \frac{f_{\Brv n}(\mby - \mbA\mbh)f_{\Brv h}^{(K)}(\mbh)}{f_{\Brv y}^{(K)}(\mby)} d \mbh
    \end{equation}
    approximates the \ac{cme}
    \begin{equation}
        \hhat(\mby) = \expec[\mbh\mid\mby] = \int \mbh \frac{f_{\Brv n}(\mby - \mbA\mbh)f_{\Brv h}(\mbh)}{f_{\Brv y}(\mby)} d \mbh
    \end{equation}
    in the sense that for any radius \( r > 0 \),
    \begin{equation}\label{eq:convergence_estimator}
        \lim_{K\to\infty} \| \hhat(\mby) - \hhat^{(K)}(\mby) \| = 0
    \end{equation}
    holds uniformly for all \( \mby \) in the ball \( \calB_r = \{ \mby \in \R^{N} : \| \mby \| \leq r \} \).
    Thus, in particular, by finding a suitable \( r > 0 \), \eqref{eq:convergence_estimator} can be seen to hold for any given \( \mby \in \R^N \).
\end{theorem}

\subsection{Discussion of Theorem 2}\label{sec:discussion_theorem2}

Both estimators \( \hhat^{(K)} \) and \( \hhat \) are functions which map the current observation \( \mby \) onto corresponding channel estimates \( \hhat^{(K)}(\mby) \) and \( \hhat(\mby) \).
\Cref{thm:main_result} proves the pointwise convergence of the function sequence \( ( \hhat^{(K)} )_{K=1}^{\infty} \) to the function \( \hhat \).
In detail, for any observation \( \mby \), the sequence \( ( \hhat^{(K)}(\mby) )_{K=1}^{\infty} \) of channel estimates converges in Euclidean norm to the optimal channel estimate \( \hhat(\mby) \).
To our knowledge, the pointwise convergence of the estimators has not been investigated yet.
The work in~\cite{Go94} can be considered to be most related to our result.
Therein, the author assumes that the random vectors \( (\mbh^{(K)}, \mby^{(K)}) \)
converge in distribution to the random vectors \( (\mbh, \mby) \) and the question
is whether the \textit{random vectors} \( \expec[\mbh^{(K)} \mid \mby^{(K)}] \)
converge in distribution to the \textit{random vector} \( \expec[\mbh \mid \mby] \).
Here, the condition is still considered to be a random vector
whereas we assume to condition on the current realization, which is given by the observation at the base station.
Thus, \cite{Go94} studies a sequence of random vectors and we study a sequence of functions.
Further, \cite{Go94} studies the convergence in distribution and we study the pointwise convergence.
Therefore, the result in \Cref{thm:main_result} is not a consequence of~\cite{Go94}.

Next, we discuss the implications of the fact that~\eqref{eq:convergence_estimator} holds uniformly for all \( \mby \) in a ball \( \calB_r \).
Let \( \calB_r \) be given.
If we want the error \( \| \hhat(\mby) - \hhat^{(K)}(\mby) \| \) to be smaller than a given threshold \( \varepsilon_{\text{thr}} > 0 \), then according to \Cref{thm:main_result}, we can find a \( K_r \in \N \) such that \( \| \hhat(\mby) - \hhat^{(K)}(\mby) \| \leq \varepsilon_{\text{thr}} \) holds for all \( \mby \in \calB_r \) and \( K \geq K_r \).
This does not mean that the error is always larger than \( \varepsilon_{\text{thr}} \) for \( \mby \notin \calB_r \).
However, it can be the case, that for certain \( \mby \notin \calB_r \) the number \( K \) of components needs to be strictly larger than \( K_r \) in order for \( \| \hhat(\mby) - \hhat^{(K)}(\mby) \| \) to fall below the threshold.

A requirement of \Cref{thm:main_result} is a sequence of \acp{pdf} which converges uniformly to \( f_{\mbh} \).
By \Cref{thm:nguyens}, such a sequence always exists if we consider \acp{gmm}.
However, as argued in~\cite{NgNgChMc20}, there exist other mixtures with universal approximation properties as well.
It is an interesting question whether~\eqref{eq:conditional_mean_K} can be computed in closed form for other mixture models and to see if they for example need fewer components for a satisfying approximation and channel estimation.

\Cref{thm:main_result} requires \( \mbA \) to be invertible.
Unfortunately, the proof of \Cref{thm:main_result} cannot be conducted as presented if \( \mbA \) is not invertible, which is for example the case when we consider a wide matrix with more columns than rows.
There are multiple challenges involved.
First, the proof makes use of \Cref{lem:bound_norm} which shows that \( \int \| \mbh \| f_{\mbn}(\mby - \mbA \mbh) d\mbh \) is finite for any \( \mby \).
This integral is generally not finite if \( \mbA \) is not invertible (see Appendix~\ref{sec:noninvertible_A}).
Second, for invertible \( \mbA \), we could directly show that the sequence of \acp{pdf} corresponding to \( \mbA \mbh^{(K)} \) converges uniformly to the \ac{pdf} of \( \mbA \mbh \).
This will likely not hold for noninvertible \( \mbA \) (see Appendix~\ref{sec:no_uniform_convergence}).

While a strong statement about the convergence of the estimators does not seem possible with the presented means if \( \mbA \) is not invertible,
we can make the following observation.
Since \( f_{\mbh}^{(K)} \) converges uniformly to \( f_{\mbh} \),
it converges in particular pointwise.
By Scheffe's lemma (e.g.,~\cite{bookRe03}), it follows that the random vectors \( \mbh^{(K)} \) converge to \( \mbh \) in distribution.
We also have \( (\mbh^{(K)}, \mbn) \to (\mbh, \mbn) \) in distribution.
If we define the continuous mapping
\begin{equation}
    s: \mathbb{R}^{N+m} \to \mathbb{R}^{N+m}, (\mbh, \mbn) \mapsto (\mbh, \mbA \mbh + \mbn)
\end{equation}
then, the continuous mapping theorem (e.g.,~\cite{bookKl08}) implies the convergence of \( s(\mbh^{(K)}, \mbn) = (\mbh^{(K)}, \mby^{(K)}) \) to
\( s(\mbh, \mbn) = (\mbh, \mby) \) in distribution.
Given the convergence in distribution of a sequence \( (\mbh^{(K)}, \mby^{(K)}) \) to \( (\mbh, \mby) \),
the author of~\cite{Go94} investigates conditions which ensure the convergence of the corresponding conditional expectations \( \expec[\mbh^{(K)} \mid \mby^{(K)}] \) to \( \expec[\mbh \mid \mby] \).
The main result depends without limitation on the distribution of \( f_{\mbh} \) which is not assumed to be known in our setting.
Reciting the main result is beyond the scope of the current paper and we refer the interested reader to~\cite{Go94}.

\section{Channel Models and Related Channel Estimators}\label{sec:sota}

Before we turn to numerical simulations, we introduce the considered channel models and discuss other channel estimation algorithms which we use for comparison.
At this point, it should be noted that the proposed approach does not rely on estimating a (link-based) covariance matrix based on pilot symbols.
Therefore, we do not show comparisons with such approaches as this lies in a different field of applications.
To generate channel samples (for training and testing purposes), we define a scenario like for example a base station which covers a certain \( 120^\circ \) sector.
Afterwards, we choose user positions uniformly at random within the scenario and retrieve their corresponding channels.
The so-obtained set of channel samples can then be used to find estimators for the whole scenario.
These estimators (both the \ac{gmm} estimator as well as all estimators introduced in the following) are trained/computed/defined once and then tested on the whole scenario without further modification.

\subsection{Channel Models}\label{sec:channel_models}

\subsubsection{3GPP}\label{sec:3gpp}

We work with a spatial channel model~\cite{3gpp,NeWiUt18} where channels are modeled conditionally Gaussian: \( \mbh \mid \mbdelta \sim \calN(\mbzero, \mbC_{\mbdelta}) \).
The random vector \( \mbdelta \) collects the angles of arrival/departure and path gains of the main propagation clusters between a mobile terminal and the base station.
The main angles are drawn independently and uniformly from the interval \( [0, 2\pi] \) and the path gains are independent zero-mean Gaussians.
The base station employs a \ac{ula} for both the transmitter and the receiver such that the transmit- and receive-side spatial channel covariance matrix are given by
\begin{equation}
    \mbC_{\mbdelta}^{\text{\{rx,tx\}}} = \int_{-\pi}^\pi g^{\text{\{rx,tx\}}}(\theta; \mbdelta) \mba^{\text{\{rx,tx\}}}(\theta) \mba^{\text{\{rx,tx\}}}(\theta)^\herm d \theta.
\end{equation}
Here,
\begin{equation}
    \mba^{\text{\{rx,tx\}}}(\theta) = [1, e^{j\pi\sin(\theta)}, \dots, e^{j\pi(N_{\text{\{rx,tx\}}}-1)\sin(\theta)}]^\tp
\end{equation}
is the array steering vector for an angle of arrival/departure \( \theta \) and \( g \) is a power density consisting of a sum of weighted Laplace densities whose standard deviations describe the angle spread of the propagation clusters~\cite{3gpp}.
The full channel covariance matrix is constructed as \( \mbC_{\mbdelta} = \mbC_{\mbdelta}^{\text{tx}} \otimes \mbC_{\mbdelta}^{\text{rx}} \) due to the assumption of independent scattering in the vicinity of transmitter and receiver, see, e.g., \cite{KeScPeMoFr02}.
In the \ac{simo} case, \( \mbC_{\mbdelta} \) degenerates to the receive-side covariance matrix \( \mbC_{\mbdelta}^{\text{rx}} \).
For every channel sample, we generate random angles and path gains, combined in \( \mbdelta \), and then draw the sample as \( \mbh \sim \calN(\mbzero, \mbC_\mbdelta) \).

\subsubsection{QuaDRiGa}\label{sec:quadriga}

Version 2.4 of the QuaDRiGa channel simulator~\cite{QuaDRiGa1, QuaDRiGa2} is used to generate channel samples.
We simulate an urban macrocell scenario at a center frequency of 2.53 GHz.
The base station's height is 25 meters and it covers a \( 120^\circ \) sector.
The minimum and maximum distances between the mobile terminals and the base station are 35 meters and 500 meters, respectively.
In 80\% of the cases, the mobile terminals are located indoors at different floor levels, whereas the mobile terminals' height is 1.5 meters in the case of outdoor locations.

QuaDRiGa models the channel of the \( c \)-th carrier and \( t \)-th time symbol as
\( \mbH_{c,t} = \sum_{\ell=1}^{L} \mbG_{\ell} e^{-2\pi j f_c \tau_{\ell,t}} \)
where \( \ell \) is the path number, and the number of multi-path components \( L \) depends on whether there is \ac{los}, \ac{nlos}, or \ac{o2i} propagation: \( L_\text{LOS} = 37 \), \( L_\text{NLOS} = 61 \) or \( L_\text{O2I} = 37 \), cf.~\cite{KuDaJaTh19}.
The frequency of the \( c \)-th carrier is denoted by \( f_c \) and the \( \ell\)-th path delay of the \( t \)-th time symbol by \( \tau_{\ell, t} \).
The coefficients matrix \( \mbG_{\ell} \) consists of one complex entry for each antenna pair, which comprises the attenuation of a path, the antenna radiation pattern weighting, and the polarization \cite{KuDaJaTh19}.
As described in the QuaDRiGa manual~\cite{QuaDRiGa2}, the generated channels are post-processed to remove the path gain.

For the simulations in \Cref{sec:simo_results} and \Cref{sec:mimo_results}, we generate single-carrier \ac{simo} and \ac{mimo} channels, respectively.
The base station is equipped with a \ac{ula} with \( \Nrx \) ``3GPP-3D'' antennas and the mobile terminals employ \( \Ntx \) ``omni-directional'' antennas.
For the simulations in \Cref{sec:wideband_results}, we consider a SISO system in the spatial domain with \( N_c \) carriers over a bandwidth of 360 kHz and for a time slot with 1 ms duration that is divided into \( N_t \) time symbols.
Each user moves with a certain velocity \( v \) in a random direction.

\subsection{State-of-the-Art Channel Estimators}\label{sec:baseline_estimators}

A simple baseline algorithm is the \ac{ls} channel estimator which computes
\begin{equation}\label{eq:ls}
    \hhat_{\text{LS}} = \mbA^\dagger \mby
\end{equation}
using the Moore-Penrose pseudoinverse \( \mbA^\dagger \).
For \( \mbA = \mbI \), there is nothing to compute, and for \( \mbA \neq \mbI \), we have a complexity of \( \calO(mN) \) because the pseudoinverse can be precomputed.
Another immediate estimator consists of first estimating a sample covariance matrix \( \mbC = \frac{1}{M} \sum_{m=1}^M \mbh_m \mbh_m^\herm \) using \( M = 10^5 \) training channel samples drawn uniformly from the whole scenario and then computing \ac{lmmse} channel estimates:
\begin{equation}\label{eq:sample_cov}
    \hhat_{\text{sample cov.}} = \mbC \mbA^\herm (\mbA \mbC \mbA^\herm + \mbSigma)^{-1} \mby.
\end{equation}
Since \( \mbC \) is computed in the offline phase, this estimator also has an online complexity of \( \calO(m N) \).
When we work with the 3GPP channel model from \Cref{sec:3gpp}, then the true covariance matrix \( \mbC_{\mbdelta} \) for every channel sample is available and we can compute a genie \ac{lmmse} channel estimate:
\begin{equation}\label{eq:genie_lmmse}
    \hhat_{\text{gen. LMMSE}} = \expec[\mbh \mid \mby, \mbdelta] = \mbC_{\mbdelta} \mbA^\herm (\mbA \mbC_{\mbdelta} \mbA^\herm + \mbSigma)^{-1} \mby
\end{equation}
which presents a lower bound for all estimators.
Note that this is not the optimal \ac{cme} considered in~\eqref{eq:conditional_mean} because of the additional genie knowledge of \( \mbdelta \).
The inverse in~\eqref{eq:genie_lmmse} needs to be computed for every observation because for every observation there is a corresponding \( \mbC_{\mbdelta} \).
Thus, the complexity is \( \calO(m^3 + mN) \).

Many modern channel estimation algorithms focus on \ac{cs} approaches,
see, e.g., the surveys~\cite{BuHuMuDa18,HaMaZwDa20}.
In what follows, we therefore consider two \ac{cs} algorithms~\cite{DaMaAv97,PaReKr93,Tr04,DoMaMo10,MaAnYaBa13}.
Recent non-\ac{cs} algorithms often focus on machine learning methods.
For this reason, we also compare with such methods in the numerical simulations~\cite{NeWiUt18,FeTuKoUt21,SoPoSh20,SoPoMiSh19}.

\ac{cs} approaches assume the channel to be (approximately) sparse: \( \mbh \approx \mbD \mbs \).
Here, \( \mbD \in \C^{N\times L} \) is a \textit{dictionary} and \( \mbs \in \C^L \) is a sparse vector.
A typical choice for \( \mbD \) is an oversampled \ac{dft} matrix (e.g., \cite{AlLeHe15}).
\ac{cs} algorithms then assume that \( \mby = \mbA \mbD \mbs + \mbn \) holds and they recover an estimate \( \hat{\mbs} \) of \( \mbs \) and estimate the channel as \( \hhat = \mbD \hat{\mbs} \).
A well-known \ac{cs} algorithm is \ac{omp} \cite{DaMaAv97,PaReKr93,Tr04}.
\ac{omp} needs to know the sparsity order.
Since order estimation is a difficult problem, we avoid it via a genie-aided approach: \ac{omp} gets access to the true channel to choose the optimal sparsity order.
This yields a performance bound for \ac{omp}.
As explained in~\cite{NeWiUt18}, every iteration has a complexity of \( \calO(N \log(N)) \)
and the number of iterations is equal to the genie-determined sparsity order.
Another algorithm, which we use for comparison, is \ac{amp} \cite{DoMaMo10,MaAnYaBa13}, which does not need to know the sparsity order.
The computational complexity of \ac{amp} is not analyzed in~\cite{DoMaMo10,MaAnYaBa13}.
Since it is an iterative algorithm, it depends without limitation on the number of iterations.
We set the number of iterations per channel estimate to 100 in the simulations.

A \ac{cnn}-based channel estimator was derived in~\cite{NeWiUt18} for the \ac{simo} signal model (cf. \Cref{sec:signal_model_simo}).
In~\cite{FeTuKoUt21}, the \ac{cnn} estimator has been generalized to the \ac{mimo} signal model (cf. \Cref{sec:signal_model_mimo}).
In the \ac{simo} case, we use the \ac{cnn} estimator as described in~\cite{NeWiUt18}.
The activation function is the rectified linear unit and we use the input transform based on the \( 2N \times 2N \) Fourier matrix, cf. \cite[Equation (43)]{NeWiUt18}.
In the \ac{mimo} case, we use the \ac{cnn} estimator as described in~\cite{FeTuKoUt21}
where again the activation function is the rectified linear unit.
In all cases, the \ac{cnn} is trained on samples corresponding to the channel model on which it is tested later.
The computational complexity is \( \calO(N \log(N)) \) \cite{NeWiUt18,FeTuKoUt21}.

The concept of a \ac{cae} was introduced in \cite{BaAbZo19} and adapted for wideband channel estimation (cf. \Cref{sec:signal_model_wideband}) in \cite{SoPoSh20}.
The \ac{cae} is an autoencoder where the encoder is replaced by a \textit{concrete selection layer} that selects the \( N_p \ll N_c N_t \) most informative features of the \( N_c N_t \)-dimensional input.
This corresponds to designing the pilot matrix by selecting the \( N_p \) pilot positions.
The decoder can then be used to perform channel estimation.
During training, noisy channels are given as input, such that the decoder of the \ac{cae} performs denoising and reconstruction of the full-dimensional channels.
Hence, a new \ac{cae} needs to be trained for every different \ac{snr}.
In our simulations, in contrast to \cite{SoPoSh20}, no further denoising networks are applied after the \ac{cae}.
The complexity of \acp{cae} is not analyzed in~\cite{SoPoSh20}.

The authors in \cite{SoPoMiSh19} propose a deep \ac{cnn} approach for 2D wideband channel estimation. 
The estimator, called \textit{ChannelNet}, consists of a combination of an image super-resolution and an image restoration network.
Thus, the networks perform interpolation and denoising of the low-dimensional observations with respect to the high-dimensional channel matrix.
The super-resolution network consists of three 2D convolution layers, whereas the image restoration network consists of 20 2D convolution layers.
Here, too, we train \ac{snr}-specific networks.
The complexity of ChannelNet is not analyzed in~\cite{SoPoMiSh19}.

\section{Numerical Simulations}\label{sec:numerical_results}

In all simulations, a normalized \ac{mse} (nMSE) is used as performance measure.
Specifically, we generate \( T = 10^4 \) \( N \)-dimensional test channel samples \( \{ \mbh_t \}_{t=1}^T \), obtain corresponding channel estimates \( \hhat_t \), and
define \( \text{nMSE} = \frac{1}{NT} \sum_{t=1}^T \| \mbh_t - \hhat_t \|^2. \)
The noise covariance matrix is \( \mbSigma = \sigma^2 \mbI \).
The test samples are normalized such that \( \expec[\|\mbh\|^2] = N \) holds which
allows us to define an \ac{snr} as \( \frac{1}{\sigma^2} \).
For training purposes, we generate \( M = 10^5 \) channel samples unless stated otherwise.
The number of training samples is always chosen large enough such that increasing \( M \) does not lead to a performance improvement during the testing phase.
The generated channel samples stem from one of the scenarios described in \Cref{sec:sota}.

As explained in \Cref{sec:gmm_channel_estimator} (see also \Cref{alg:gmm_estimator} there), to obtain the \ac{gmm}-based estimator,
we fit one \ac{gmm} using the available training data via an \ac{em} algorithm.
Afterwards, inverses which appear in~\eqref{eq:estimator_full} are precomputed for every \ac{snr}.
In contrast to this approach, the introduced neural network-based estimators need to be newly trained for every \ac{snr} at which we evaluate them.
This includes searching for suitable hyperparameters for every \ac{snr}.

\subsection{SIMO}\label{sec:simo_results}

\begin{figure}[t]
	\centering
	\begin{tikzpicture}
		\begin{axis}
			[width=\plotwidth,
			height=\plotheightSimoComponents,
			xtick=data,
			xmin=16, 
			xmax=128,
			xlabel={\( N \), number of antennas},
			ymode = log, 
			ymin= 1e-2,
			ymax=1.1e-1,
			ylabel= {Normalized MSE}, 
			ylabel shift = 0.0cm,
			grid = both,
			legend columns = 3,
			legend entries={
				\legendAmp,
				\legendCnn,
				\legendGenielmmse,
				\legendGlobalcov,
				\legendGmm,
				\legendGmmDiag,
				\legendLs,
				\legendOmp,
			},
			legend style={at={(0.0,1.0)}, anchor=south west},
			legend cell align = {left},
			]

			\addplot[amp]
			table[x=antennas, y=amp, col sep=comma]
			{csv/sim_1_to_5_different_nantennas.csv};

			\addplot[cnn]
			table[x=antennas, y=cnn_fft2x_relu_non_hier_False, col sep=comma]
			{csv/2022_05_05_12-46-08_1paths_20lbatch_20lsize_500ebatch_snr=10_new.csv};

			\addplot[genielmmse]
			table[x=antennas, y=genie lmmse, col sep=comma]
			{csv/sim_1_to_5_different_nantennas.csv};

			\addplot[globalcov]
			table[x=antennas, y=global cov, col sep=comma]
			{csv/sim_1_to_5_different_nantennas.csv};

			\addplot[gmm]
			table[x=antennas, y=gmm full all, col sep=comma]
			{csv/sim_1_to_5_different_nantennas.csv};
			
			\addplot[gmmdiag]
			table[x=antennas, y=gmm diag all, col sep=comma]
			{csv/sim_1_to_5_different_nantennas.csv};
			\label{curve:gmmdiag}
			
			\addplot[ls]
			table[x=antennas, y=ls, col sep=comma]
			{csv/sim_1_to_5_different_nantennas.csv};
			
			\addplot[omp]
			table[x=antennas, y=Genie_OMP, col sep=comma]
			{csv/2022_05_05_12-46-08_1paths_20lbatch_20lsize_500ebatch_snr=10_new.csv};
			
		\end{axis}
	\end{tikzpicture}
	\caption{SIMO signal model (\Cref{sec:signal_model_simo}) and 3GPP channel model (\Cref{sec:3gpp}) with \textbf{one} propagation cluster at 10 dB \ac{snr}. The performance of the circulant \ac{gmm} estimator (``circ. GMM'', \ref{curve:gmmdiag}, \Cref{sec:circulant_covs}) is shown too. In both cases, \( K = 128 \) components are used.}
	\label{fig:1path_simo_vs_antennas}
\end{figure}
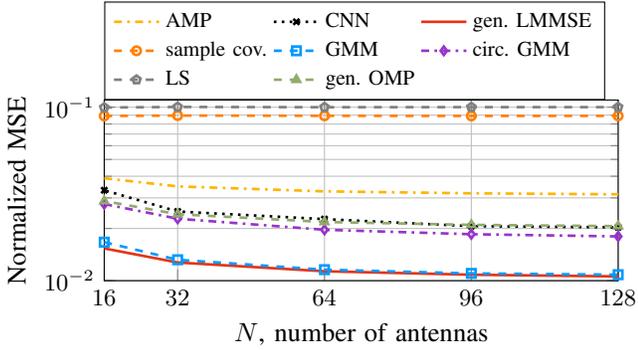

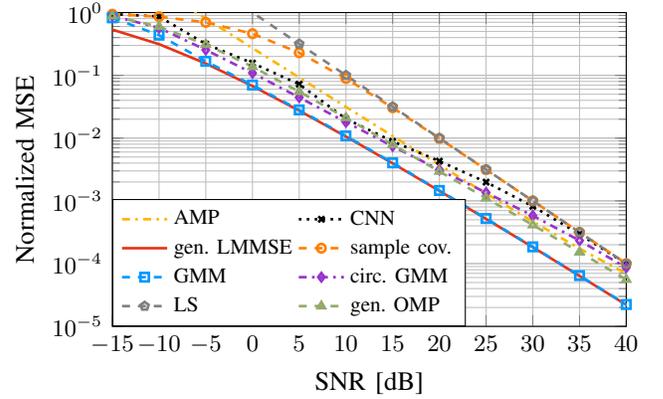
\begin{figure}[t]
	\centering
	\begin{tikzpicture}
		\begin{axis}
			[width=\plotwidth,
			height=\plotheightSimoOnepath,
			xtick=data,
			xmin=-15, 
			xmax=40,
			xlabel={SNR [dB]},
			ymode = log, 
			ymin= 1e-5,
			ymax=1,
			ylabel= {Normalized MSE}, 
			ylabel shift = 0.0cm,
			grid = both,
			legend columns = 2,
			legend entries={
				\legendAmp,
				\legendCnn,
				\legendGenielmmse,
				\legendGlobalcov,
				\legendGmm,
				\legendGmmDiag,
				\legendLs,
				\legendOmp,
			},
			legend style={at={(0.0,0.0)}, anchor=south west},
			legend cell align = {left},
			]

			\addplot[amp]
			table[x=SNR, y=amp, col sep=comma]
			{csv/sim_55_results.csv};

			\addplot[cnn]
			table[x=SNR, y=cnn_fft2x_relu_non_hier_False, col sep=comma]
			{csv/cnn_1paths_128antennas_20lbatch_20lsize_500ebatch_2sigma_update.csv};

			\addplot[genielmmse]
			table[x=SNR, y=genie lmmse, col sep=comma]
			{csv/sim_55_results.csv};

			\addplot[globalcov]
			table[x=SNR, y=global cov, col sep=comma]
			{csv/sim_55_results.csv};

			\addplot[gmm]
			table[x=SNR, y=gmm full all, col sep=comma]
			{csv/sim_55_results.csv};
			
			\addplot[gmmdiag]
			table[x=SNR, y=gmm diag all, col sep=comma]
			{csv/sim_55_results.csv};
			
			\addplot[ls]
			table[x=SNR, y=ls, col sep=comma]
			{csv/sim_55_results.csv};
			
			\addplot[omp]
			table[x=SNR, y=Genie_OMP, col sep=comma]
			{csv/cnn_1paths_128antennas_20lbatch_20lsize_500ebatch_2sigma_update.csv};
			
		\end{axis}
	\end{tikzpicture}
	\caption{SIMO signal model (\Cref{sec:signal_model_simo}) and 3GPP channel model (\Cref{sec:3gpp}) with \textbf{one} propagation cluster and \( N = 128 \) antennas. The performance of the circulant \ac{gmm} estimator (``circ. GMM'', \ref{curve:gmmdiag}, \Cref{sec:circulant_covs}) is shown too. In both cases, \( K = 128 \) components are used.}
	\label{fig:1path_simo}
\end{figure}

\begin{figure}[t]
	\centering
	\begin{tikzpicture}
		\begin{axis}
			[width=\plotwidth,
			height=\plotheightSimoThreepath,
			xtick=data,
			xmin=-15, 
			xmax=40,
			xlabel={SNR [dB]},
			ymode = log, 
			ymin= 5e-5,
			ymax=1,
			ylabel= {Normalized MSE}, 
			ylabel shift = 0.0cm,
			grid = both,
			legend columns = 2,
			legend entries={
				\legendAmp,
				\legendCnn,
				\legendGenielmmse,
				\legendGlobalcov,
				\legendGmm,
				\legendGmmDiag,
				\legendLs,
				\legendOmp,
			},
			legend style={at={(0.0,0.0)}, anchor=south west},
			legend cell align = {left},
			]

			\addplot[amp]
			table[x=SNR, y=amp, col sep=comma]
			{csv/sim_56_results.csv};

			\addplot[cnn]
			table[x=SNR, y=cnn_fft2x_relu_non_hier_False, col sep=comma]
			{csv/cnn_3paths_128antennas_20lbatch_20lsize_500ebatch_2sigma.csv};

			\addplot[genielmmse]
			table[x=SNR, y=genie lmmse, col sep=comma]
			{csv/sim_56_results.csv};

			\addplot[globalcov]
			table[x=SNR, y=global cov, col sep=comma]
			{csv/sim_56_results.csv};

			\addplot[gmm]
			table[x=SNR, y=gmm full all, col sep=comma]
			{csv/sim_56_results.csv};
			
			\addplot[gmmdiag]
			table[x=SNR, y=gmm diag all, col sep=comma]
			{csv/sim_56_results.csv};
			
			\addplot[ls]
			table[x=SNR, y=ls, col sep=comma]
			{csv/sim_56_results.csv};
			
			\addplot[omp]
			table[x=SNR, y=Genie_OMP, col sep=comma]
			{csv/cnn_3paths_128antennas_20lbatch_20lsize_500ebatch_2sigma.csv};
			
		\end{axis}
	\end{tikzpicture}
	\caption{SIMO signal model (\Cref{sec:signal_model_simo}) and 3GPP channel model (\Cref{sec:3gpp}) with \textbf{three} propagation clusters and \( N = 128 \) antennas. The performance of the circulant \ac{gmm} estimator (``circ. GMM'', \ref{curve:gmmdiag}, \Cref{sec:circulant_covs}) is shown too. In both cases, \( K = 128 \) components are used.}
	\label{fig:3paths_simo}
	\vspace{-2mm}
\end{figure}
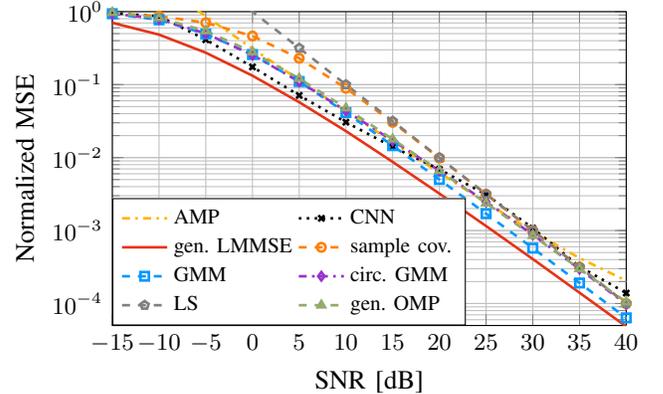

\begin{figure}[t]
	\centering
	\begin{tikzpicture}
		\begin{axis}
			[width=\plotwidth,
			height=\plotheightSimoQuadriga,
			xtick=data,
			xmin=-15, 
			xmax=40,
			xlabel={SNR [dB]},
			ymode = log, 
			ymin= 1e-4,
			ytick={1e-3,1e-2,1e-1},
			ymax=1,
			ylabel= {Normalized MSE}, 
			ylabel shift = 0.0cm,
			grid = both,
			legend columns = 2,
			legend entries={
				\legendAmp,
				\legendCnn,
				\legendGlobalcov,
				\legendGmm,
				\legendGmmDiag,
				\legendLs,
				\legendOmp,
			},
			legend style={at={(0.0,0.0)}, anchor=south west},
			legend cell align = {left},
			]

			\addplot[amp]
			table[x=SNR, y=amp, col sep=comma]
			{csv/sim_74_results.csv};

			\addplot[cnn]
			table[x=SNR, y=cnn_fft2x_relu_non_hier_False, col sep=comma]
			{csv/cnn_128antennas_quadriga.csv};

			\addplot[globalcov]
			table[x=SNR, y=global cov, col sep=comma]
			{csv/sim_74_results.csv};
			
			\addplot[gmm]
			table[x=SNR, y=gmm full all, col sep=comma]
			{csv/sim_74_results.csv};
			
			\addplot[gmmdiag]
			table[x=SNR, y=gmm diag all, col sep=comma]
			{csv/sim_74_results.csv};
			
			\addplot[ls]
			table[x=SNR, y=ls, col sep=comma]
			{csv/sim_74_results.csv};
			
			\addplot[omp]
			table[x=SNR, y=Genie_OMP, col sep=comma]
			{csv/cnn_128antennas_quadriga.csv};
			
		\end{axis}
	\end{tikzpicture}
	\caption{SIMO signal model (\Cref{sec:signal_model_simo}) and QuaDRiGa channel model (\Cref{sec:quadriga}) with \( N = 128 \) antennas. The performance of the circulant \ac{gmm} estimator (``circ. GMM'', \ref{curve:gmmdiag}, \Cref{sec:circulant_covs}) is shown too. In both cases, \( K = 128 \) components are used.}
	\label{fig:quadriga_simo}
\end{figure}
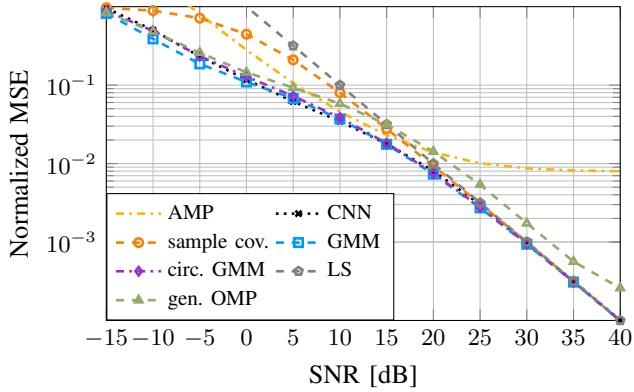

\Cref{fig:1path_simo_vs_antennas,fig:1path_simo,fig:3paths_simo,fig:quadriga_simo,fig:simo_components} show channel estimation results for the \ac{simo} signal model from \Cref{sec:signal_model_simo}.
The \ac{cs} algorithms \ac{omp} and \ac{amp} are used with oversampled \ac{dft} dictionaries that have \( L = 4 N \) and \( L = 2 N \) columns, respectively,
because these parameters yielded the best results.
Unless stated otherwise, the \ac{gmm} fitting process uses \( 19 \cdot 10^4 \) training data.

In \Cref{fig:1path_simo_vs_antennas}, we consider the 3GPP channel model from \Cref{sec:3gpp} with one propagation cluster.
The \ac{snr} is 10 dB and the number \( N \) of antennas is varied.
It is interesting to see that the \ac{gmm}-based estimator performs almost as well as the genie \ac{lmmse} estimator.
As the number \( N \) of antennas increases, the \ac{cnn} estimator starts to outperform the genie \ac{omp} estimator.
The reason for this is that the assumptions under which the \ac{cnn} estimator was derived in~\cite{NeWiUt18} are better fulfilled for a larger number of antennas.
Generally, the relative performance between all estimators hardly differs with different numbers of antennas.
This is an observation we have made in all our experiments.
For this reason, in what follows, the number of antennas is fixed at \( N = 128 \).

In \Cref{fig:1path_simo}, we consider again the 3GPP channel model from \Cref{sec:3gpp} with one propagation cluster.
For almost all \ac{snr} values, the \ac{gmm}-based estimator performs almost as well as the genie \ac{lmmse} estimator.
In the mid-\ac{snr} range, the two \ac{cs} algorithms are approximately equally good.
In \Cref{fig:3paths_simo}, we have three propagation clusters.
A first observation is the strong performance of the \ac{cnn} estimator in the mid-\ac{snr} range.
Note that we can generally not expect any estimator to reach the genie \ac{lmmse} curve because it has more channel knowledge (the true covariance matrix for every sample).
In the higher \ac{snr}-range, the \ac{gmm}-based estimator is the only algorithm still outperforming \ac{ls} estimation.

In \Cref{fig:quadriga_simo}, we concentrate on the QuaDRiGa channel model described in \Cref{sec:quadriga} where the channel covariance matrices and therefore the genie \ac{lmmse} curve are no longer available.
Here, the two \ac{cs} algorithms behave not as similarly as they did in the previous experiments.
Additionally, their performance is not as convincing.
A reason might be that the channels now are not sparse enough.
The \ac{cnn} estimator shows again a good performance and overall the \ac{gmm}-based estimator can compete with it or is better.

In addition to the \ac{gmm} estimator~\eqref{eq:estimator_full},
\Cref{fig:1path_simo_vs_antennas,fig:1path_simo,fig:3paths_simo,fig:quadriga_simo} display the performance of the reduced-complexity \ac{gmm} estimator which uses circulant covariance matrices as described in \Cref{sec:circulant_covs}.
As expected, the estimator's performance suffers but it is still comparable to the other algorithms.
\Cref{fig:quadriga_simo} is particularly interesting where there is not much difference between the full- and low-complexity \ac{gmm} estimators.

\Cref{fig:simo_components} shows the behavior of the \ac{gmm}-based estimator for different numbers of components, \( K \).
We consider an \ac{snr} of 10 dB and the same situation as in \Cref{fig:3paths_simo}:
The 3GPP channel model (cf.~\Cref{sec:3gpp}) with three propagation clusters and \( N = 128 \) antennas.
In addition to \( K \), also the number of training data used to fit the \ac{gmm} is varied.
Since the number of parameters of a \ac{gmm} increases when \( K \) is increased,
more training data is necessary for a good fit.
This effect is clearly visible in \Cref{fig:simo_components}.
Overall, as long as the number of training data is high enough (\( M \geq 300 \) in the figure),
increasing \( K \) leads to an \ac{mse} improvement,
which is in accordance with \Cref{thm:main_result}.

Note that we cannot expect the \ac{gmm} estimator to converge to the genie \ac{lmmse} estimator~\eqref{eq:genie_lmmse} (which is displayed in \Cref{fig:3paths_simo}).
The genie \ac{lmmse} estimator has more knowledge (namely the true channel covariance matrix \( \mbC_{\mbdelta} \)) and is therefore not the \ac{cme}, \( \hhat = \expec[\mbh \mid \mby] \), which we want to approximate in \Cref{thm:main_result}.
The \ac{cme} cannot be computed in closed form in the considered scenario which is the main motivation to study the \ac{gmm} estimator in the first place.

\Cref{fig:quadriga_simo_K} also shows the \ac{gmm} estimator's behavior for different numbers of components, \( K \),
but now for the QuaDRiGa channel model (cf.~\Cref{sec:quadriga}).
For all displayed \acp{snr}, a saturation can be observed as \( K \) is increased,
and already a moderate number of components can lead to a satisfactory estimation performance.
Altogether, a smaller \( K \) tends to be sufficient for higher \acp{snr}.
Generally, a suitable number of components needs to be determined based on the training data size as well as on the desired estimator complexity.

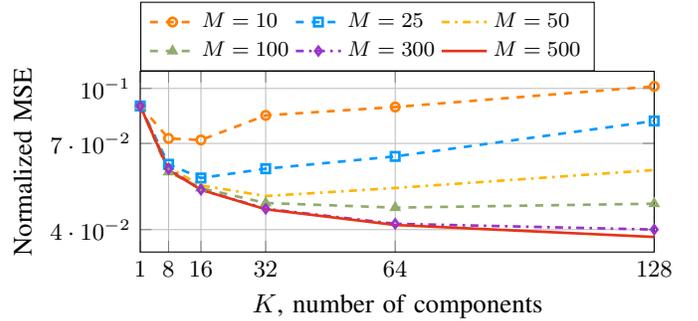
\begin{figure}[t]
	\centering
	\begin{tikzpicture}
		\begin{axis}
			[width=\plotwidth,
			height=\plotheightSimoComponents,
			xtick=data,
			xlabel={\( K \), number of components},
			ymode = log,
			xmin=1,
			xmax=128,
			ytick={4e-2,7e-2,1e-1},
			yticklabels={\( 4 \cdot 10^{-2} \), \( 7 \cdot 10^{-2} \), \( 10^{-1} \)},
			ylabel={Normalized MSE},
			ylabel shift = 0.0cm,
			grid = both,
			legend columns = 3,
			legend entries={
                \footnotesize \( M = 10 \),
                \footnotesize \( M = 25 \),
                \footnotesize \( M = 50 \),
                \footnotesize \( M = 100 \),
                \footnotesize \( M = 300 \),
                \footnotesize \( M = 500 \),
			},
			legend cell align = {left},
			legend style={at={(0.0,1.0)}, anchor=south west},
			]

			\addplot[globalcov]
			table[x=components, y=snr_10, col sep=comma, discard if not={ntrain}{10}]
			{csv/case_6_components_ntrain.csv};

			\addplot[gmm]
			table[x=components, y=snr_10, col sep=comma, discard if not={ntrain}{25}]
			{csv/case_6_components_ntrain.csv};

			\addplot[amp]
			table[x=components, y=snr_10, col sep=comma, discard if not={ntrain}{50}]
			{csv/case_6_components_ntrain.csv};

			\addplot[omp]
			table[x=components, y=snr_10, col sep=comma, discard if not={ntrain}{100}]
			{csv/case_6_components_ntrain.csv};

			\addplot[gmmdiag]
			table[x=components, y=snr_10, col sep=comma, discard if not={ntrain}{300}]
			{csv/case_6_components_ntrain.csv};

			\addplot[genielmmse]
			table[x=components, y=snr_10, col sep=comma, discard if not={ntrain}{500}]
			{csv/case_6_components_ntrain.csv};

		\end{axis}
	\end{tikzpicture}
	\caption{SIMO signal model (\Cref{sec:signal_model_simo}) and 3GPP channel model (\Cref{sec:3gpp}) with \textbf{three} propagation clusters and \( N = 128 \) antennas. The SNR is 10 dB. The GMM estimator is trained using \( M \cdot 10^3 \) samples.}
	\label{fig:simo_components}
\end{figure}

\begin{figure}[t]
	\centering
	\begin{tikzpicture}
		\begin{axis}
			[width=\plotwidth,
			height=\plotheightSimoOnepath,
			xtick={1,2,4,8,16,32,64,128},
			xticklabels={1,,,8,16,32,64,128},
			xmin=1, 
			xmax=128,
			xlabel={\( K \), number of components},
			ymode = log, 
			ymin=1e-3,
			ymax=1.0,
			ylabel = {Normalized MSE}, 
			ylabel shift = 0.0cm,
			grid = both,
			legend columns = 2,
			legend entries={
				\( \text{SNR} = -10 \text{ dB} \),
				\( \text{SNR} = 0 \text{ dB} \),
				\( \text{SNR} = 10 \text{ dB} \),
				\( \text{SNR} = 20 \text{ dB} \),
			},
			legend style={at={(0.0,0.0)}, anchor=south west},
			legend cell align = {left},
			]

			\addplot[gmm]
			table[x=K, y=nMSE, col sep=comma, discard if not={SNR}{-10}]
			{csv/quadriga_K.csv};

			\addplot[omp]
			table[x=K, y=nMSE, col sep=comma, discard if not={SNR}{0}]
			{csv/quadriga_K.csv};

			\addplot[cnn]
			table[x=K, y=nMSE, col sep=comma, discard if not={SNR}{10}]
			{csv/quadriga_K.csv};

			\addplot[gmmdiag]
			table[x=K, y=nMSE, col sep=comma, discard if not={SNR}{20}]
			{csv/quadriga_K.csv};
  
		\end{axis}
	\end{tikzpicture}
	\caption{SIMO signal model (\Cref{sec:signal_model_simo}) and QuaDRiGa channel model (\Cref{sec:quadriga}) with \( N = 128 \) antennas. The training data size is \( 19 \cdot 10^4 \).}
	\label{fig:quadriga_simo_K}
\end{figure}
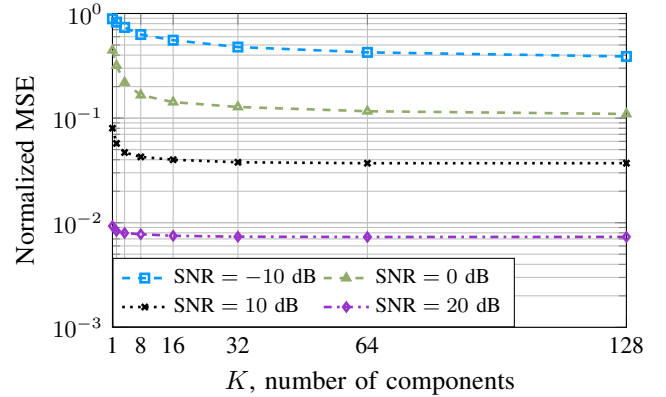

\subsection{MIMO}\label{sec:mimo_results}

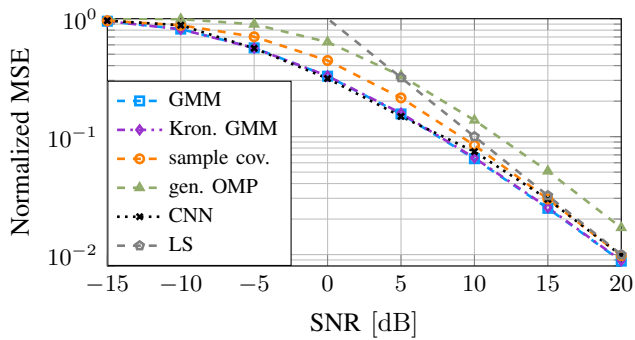
\begin{figure}[t]
	\centering
	\begin{tikzpicture}
		\begin{axis}
			[width=\plotwidth,
			height=\plotheightMimo,
			xtick=data, 
			xmin=-15, 
			xmax=20,
			xlabel={SNR \( [\operatorname{dB}] \)},
			ymode = log, 
			ymin= 8*1e-3,
			ymax=1,
			ylabel= {Normalized MSE}, 
			ylabel shift = 0.0cm,
			grid = both,
			legend columns = 1,
			legend entries={
				\legendGmm,
				\legendGmmKron,
				\legendGlobalcov,
				\legendOmp,
				\legendCnn,
				\legendLs,
			},
			legend style={at={(0.0,0.0)}, anchor=south west},
			legend cell align = {left},
			]
			
			\addplot[gmm]
			table[x= SNR, y=gmm_full, col sep=comma]
			{csvdat/MIMO_GMM/2021-08-24_19-36-48_tx=4_rx=32_train=100000.csv};
			\label{curve:gmm}
			
			\addplot[gmmdiag]
			table[x= SNR, y=gmm_kron, col sep=comma]
			{csvdat/MIMO_GMM/2021-08-24_19-36-48_tx=4_rx=32_train=100000.csv};
			
			\addplot[globalcov]
			table[x= SNR, y=Genie_Aided, col sep=comma]
			{csvdat/MIMO_GMM/2021-08-24_15-41-19_3paths_32BS_antennas_4MS_antennas_35AS_1epochs_100ebatch_20lbatchsize_4pilots.csv};
			
			\addplot[omp]
			table[x= SNR, y=Genie_OMP, col sep=comma]
			{csvdat/MIMO_GMM/2021-08-24_15-41-19_3paths_32BS_antennas_4MS_antennas_35AS_1epochs_100ebatch_20lbatchsize_4pilots.csv};
			
			\addplot[cnn]
			table[x= SNR, y=cnn_fft_relu, col sep=comma]
			{csvdat/MIMO_GMM/2021-08-26_07-36-55_3paths_32BS_antennas_4MS_antennas_35AS_1epochs_100ebatch_20lbatchsize_4pilots.csv};
			
			\addplot[ls]
			table[x=SNR, y=ls, col sep=comma]
			{csv/sim_55_results.csv};
			
		\end{axis}
	\end{tikzpicture}
	\caption{MIMO signal model (\Cref{sec:signal_model_mimo}) and QuaDRiGa channel model (\Cref{sec:quadriga}) with \( (\Nrx, \Ntx) = (32, 4) \). The GMM estimator (``GMM'', \ref{curve:gmm}) uses \( K = 32 \) components  and the Kronecker GMM estimator (``Kron. GMM'', \ref{curve:gmmdiag}, \Cref{sec:kronecker_covs}) uses \( K = 4 \times 8 \) components.}
	\label{fig:quadriga_mimo}
\end{figure}

\begin{figure}[t]
	\centering
	\begin{tikzpicture}
		\begin{axis}
			[width=\plotwidth,
			height=\plotheightMimoTrainsamples,
			xtick=data, 
			xlabel= {\( M \), number of training samples},
			ymode = log, 
			xmode=log,
			xmin=100,
			xmax=100000,
			ymin= 1e-2,
			ymax=1,
			ylabel= {Normalized MSE}, 
			ylabel shift = 0.0cm,
			grid = both,
			legend columns = 1,
			legend entries={
			    \legendGmm,
				\legendGmmKron,
			},
			legend style={at={(1.0,1.0)}, anchor=north east},
			legend cell align = {left},
			]
			
			\addplot[gmm]
			table[x=num, y=gmm_full, col sep=comma]
			{csvdat/MIMO_GMM/2021-08-24_19-36-48_tx=4_rx=32_train_snr=10db.csv};
			
			\addplot[gmmdiag]
			table[x=num, y=gmm_kron, col sep=comma]
			{csvdat/MIMO_GMM/2021-08-24_19-36-48_tx=4_rx=32_train_snr=10db.csv};

		\end{axis}
	\end{tikzpicture}
	\caption{MIMO signal model (\Cref{sec:signal_model_mimo}) and QuaDRiGa channel model (\Cref{sec:quadriga}) with \( (\Nrx, \Ntx) = (32, 4) \). The SNR is 10 dB. The Kronecker product GMM estimator (``Kron. GMM'', \ref{curve:gmmdiag}, \Cref{sec:kronecker_covs}) has fewer parameters than the normal GMM estimator. Both have \( K = 32 \) components.}
	\label{fig:mimo_train_samples}
\end{figure}
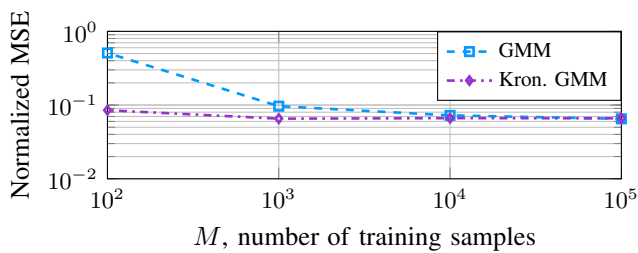

For the \ac{mimo} simulations whose signal model is described in \Cref{sec:signal_model_mimo},
we use a scaled \ac{dft} pilot matrix \( \mbP \).
Here, we are mainly interested in comparing the \ac{gmm} estimator in~\eqref{eq:estimator_full} to a \ac{gmm} estimator which uses Kronecker product covariance matrices as described in 
\Cref{sec:kronecker_covs}.
We generate \( M = 10^5 \) training channel samples \( \{ \mbH_i \}_{i=1}^M \) and use them in two different ways.
Either, we use the (vectorized) channel matrices \( \mbH_i \) directly to fit a single \( K = 32 \) components \ac{gmm}.
Or, we view all columns and all rows of the data as separate data sets and fit two \acp{gmm}:
One transmit side \ac{gmm} with \( \Ktx = 4 \) components using the rows,
one receive side \ac{gmm} with \( \Krx = 8 \) components using the columns.
Afterwards, we combine the two \acp{gmm} to a single full-size \ac{gmm} with \( K = \Ktx \Krx = 32 \) components for the whole data set of channel matrices.
To this end, the \( K \) full-size means and covariances are directly calculated by combinatorial computation of the Kronecker products of the transmit- and receive-side \ac{gmm} components.
The corresponding mixing coefficients can be computed by fixing the means and covariances and performing a single E-step (cf. \cite{bookBi06}) in the \ac{em} algorithm to obtain them.

\Cref{fig:quadriga_mimo} displays simulation results.
A first observation is that the two \ac{gmm} estimators as well as the \ac{cnn} estimator perform very similarly, with minor differences in the lower and higher \ac{snr}-regimes.
Further, these three estimators outperform the sample covariance matrix-based estimator and the genie-aided \ac{omp} algorithm, which uses a Kronecker product of two two-times oversampled \ac{dft} matrices as dictionary.
A second observation is that there is almost no difference between the \ac{gmm} estimator with or without Kronecker product covariance matrices.
This is insofar surprising as the normal \ac{gmm} consists of \( K = 32 \) covariance matrices of dimension \( N \times N \) with \( N = 32 \cdot 4 = 128 \)
which means that it has \( K \frac{N(N+1)}{2} = 264192 \) covariance parameters,
whereas in contrast, the Kronecker \ac{gmm} has only \( \Krx \frac{\Nrx(\Nrx+1)}{2} + \Ktx \frac{\Ntx(\Ntx+1)}{2} = 4224 + 40 = 4264 \) covariance parameters.
Since the Kronecker \ac{gmm} has significantly fewer parameters,
it should require a smaller number \( M \) of training data.
This is confirmed in \Cref{fig:mimo_train_samples} where the two estimators are compared at an \ac{snr} of 10 dB for varying \( M \).
The curves intersect between \( M = 10^4 \) and \( M = 10^5 \).

\subsection{Wideband}\label{sec:wideband_results}

In this section, we show numerical results for the wideband signal model described in Section \ref{sec:signal_model_wideband} and with the QuaDRiGa simulation setup described in Section \ref{sec:quadriga}.
We chose a typical 5G frame structure as defined in \cite{3GPP_5G} with \( N_c = 24 \) carriers over a bandwidth of 360 kHz with 15 kHz carrier spacing and with \( N_t = 14 \) time symbols over a time slot with 1 ms duration.
The number of pilot symbols is \( N_p = 50 \), which means that 50 out of the \( 24 \times 14 = 336 \) resource elements are occupied with pilot tones.
We compare the \ac{gmm}-based estimator with the \ac{lmmse} estimator based on the sample covariance matrix (see~\eqref{eq:sample_cov}), with the \ac{cae} approach, and with the ChannelNet estimator (see \Cref{sec:baseline_estimators}).
The training data for each approach consist of \( M = 10^5 \) channel realizations and corresponding observations from the pilot positions.

In \Cref{fig:wideband1}, we depict \ac{mse} results over the \ac{snr} for a scenario where every user moves at \( v=3 \) km/h speed, in which the block-type pilot arrangement has shown the best results.
One can observe superior performance of the \ac{cae} over the sample covariance \ac{lmmse} estimator and the ChannelNet, which may be due to the fact that the \ac{cae} optimizes the pilot pattern.
However, the \ac{gmm} approach is able to outperform all baseline algorithms, where the performance gap increases with increasing \ac{snr}.
Further, the impact of more components (from \( K=8 \) to \( K=128 \)) is visible and results in better performance for all \ac{snr} values.

\Cref{fig:wideband2} shows the same setup but now each user's velocity is randomly chosen between 0 - 300 km/h (in both training and testing sets), which makes the estimation more challenging and the lattice-type pilot arrangement superior.
First, the performance gap to the sample covariance \ac{lmmse} estimator increases, which is a result of the more diverse setting that cannot be captured well by a single covariance matrix.
Further, the number of \ac{gmm} components seems to play a more important role.
The \ac{gmm} with \( K=128 \) components is still able to compete with the ChannelNet and \ac{cae} approaches.

Finally, \Cref{fig:wideband3} shows the \ac{mse} behavior for different numbers of components, \( K \), and for different amounts of training data used to fit the \ac{gmm}.
Similar to the results in \Cref{fig:simo_components}, also the wideband case shows an improving performance for increasing numbers of components.
This at least provides numerical evaluation of the convergence of the \ac{gmm} estimator for noninvertible observation matrices.

\begin{figure}[t]
	\centering
	\begin{tikzpicture}
		\begin{axis}
			[width=\plotwidth,
			height=\plotheightWidebandFixedSpeed,
			xtick=data, 
			xmin=-15, 
			xmax=40,
			xlabel={SNR $[\operatorname{dB}]$},
			ymode = log, 
			ymin= 1e-5,
			ymax=1,
			ylabel= {Normalized MSE}, 
			ylabel shift = 0.0cm,
			grid = both,
			legend columns = 1,
			legend entries={
				\legendGlobalcov \, (block),
				\legendChannelnet \, (block),
				\legendCae,
				\legendGmm \, (block); \( K = 8 \),
				\legendGmm \, (block); \( K = 128 \),
			},
			legend style={at={(0.0,0.0)}, anchor=south west},
			legend cell align = {left},
			]

			\addplot[globalcov]
			table[x= SNR, y=MMSE_block, col sep=comma]
			{csvdat/CAE_chEst/2021-07-19_16-07-10_carrier=24_symbols=14_pilots=50_epochs=10_cluster=8_v=3kmh.csv};
			
			\addplot[channelnet]
			table[x= SNR, y=channel_net, col sep=comma]
			{csvdat/ChannelNet/2021-07-19_15-00-35_carrier=24_symbols=14_pilots=50_v=3kmh_debug_mode=False_type=block.csv};

			\addplot[cae]
			table[x= SNR, y=CAE, col sep=comma]
			{csvdat/CAE_chEst/2021-07-26_05-51-24_carrier=24_symbols=14_pilots=50_epochs=300_cluster=8_v=3kmh_diff.csv};

			\addplot[gmm]
			table[x= SNR, y=gmm_from_y_block, col sep=comma]
			{csvdat/GMM/2021-07-21_11-25-55_carrier=24_symbols=14_pilots=50_epochs=300_cluster=8_v=3kmh.csv};
			
			\addplot[gmmdiag]
			table[x= SNR, y=gmm_from_y_block, col sep=comma]
			{csvdat/GMM/2021-09-21_06-16-52_carrier=24_symbols=14_pilots=50_epochs=2_cluster=128_v=3kmh_diff.csv};
			
		\end{axis}
	\end{tikzpicture}
	\caption{Wideband signal model (\Cref{sec:signal_model_wideband}) and QuaDRiGa channel model (\Cref{sec:quadriga})
	with \( N_p = 50 \) pilot tones for \( N_c = 24 \) carriers and \( N_t = 14 \) time symbols for \( v = 3 \) km/h.}
	\label{fig:wideband1}
\end{figure}
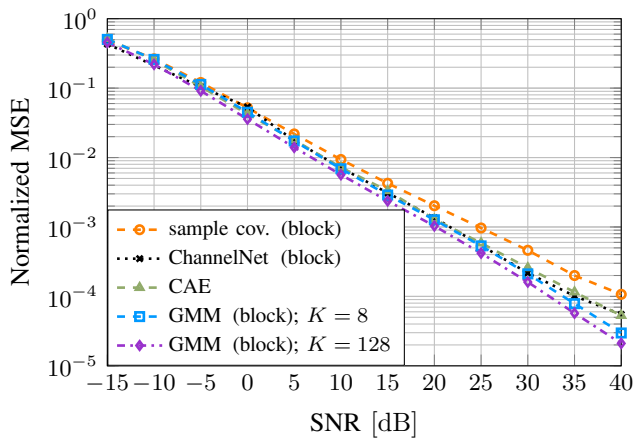

\begin{figure}[t]
	\centering
	\begin{tikzpicture}
		\begin{axis}
			[width=\plotwidth,
			height=\plotheightWidebandVariousSpeed,
			xtick=data, 
			xmin=-15, 
			xmax=40,
			xlabel={SNR $[\operatorname{dB}]$},
			ymode = log, 
			ymin= 1e-4,
			ymax=1,
			ylabel= {Normalized MSE}, 
			ylabel shift = 0.0cm,
			grid = both,
			legend columns = 1,
			legend entries = {
			    \legendGlobalcov \, (lattice),
			    \legendChannelnet \, (lattice),
			    \legendCae,
			    \legendGmm \, (lattice); \( K = 8 \),
			    \legendGmm \, (lattice); \( K = 128 \),
			},
			legend style={at={(0.0,0.0)}, anchor=south west},
			legend cell align = {left},
			]
			
			\addplot[globalcov]
			table[x= SNR, y=MMSE_lattice, col sep=comma]
			{csvdat/CAE_chEst/2021-09-03_22-31-32_carrier=24_symbols=14_pilots=50_epochs=100_cluster=64_v=300kmh_diff.csv};
			
			\addplot[channelnet]
			table[x= SNR, y=channel_net_lattice, col sep=comma]
			{csvdat/CAE_chEst/2021-09-08_10-51-37_carrier=24_symbols=14_pilots=50_epochs=100_cluster=128_v=300kmh_diff.csv};
			
    		\addplot[cae]
			table[x= SNR, y=CAE, col sep=comma]
			{csvdat/CAE_chEst/2021-09-04_22-57-48_carrier=24_symbols=14_pilots=50_epochs=100_cluster=16_v=300kmh_diff.csv};
			
			\addplot[gmm]
			table[x= SNR, y=gmm_from_y_lattice, col sep=comma]
			{csvdat/CAE_chEst/2021-09-04_08-27-44_carrier=24_symbols=14_pilots=50_epochs=100_cluster=8_v=300kmh_diff.csv};
			
			\addplot[gmmdiag]
			table[x= SNR, y=gmm_from_y_lattice, col sep=comma]
			{csvdat/CAE_chEst/2021-09-04_12-51-56_carrier=24_symbols=14_pilots=50_epochs=100_cluster=128_v=300kmh_diff.csv};
		
		\end{axis}
	\end{tikzpicture}
	\caption{Wideband signal model (\Cref{sec:signal_model_wideband}) and QuaDRiGa channel model (\Cref{sec:quadriga}) with \( N_p=50 \) pilot tones for \( N_c=24 \) carriers and \( N_t=14 \) time symbols for \( v \in [0,300] \) km/h.}
	\label{fig:wideband2}
\end{figure}
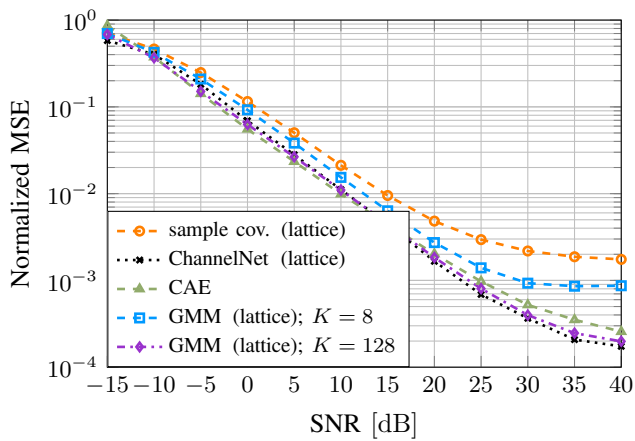

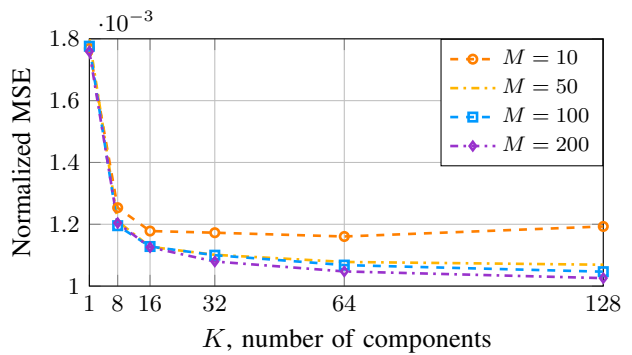
\begin{figure}[t]
	\centering
	\begin{tikzpicture}
		\begin{axis}
			[width=\plotwidth,
			height=\plotheightWidebandComponents,
			xtick=data, 
			xmin=1, 
			xmax=128,
			xlabel={\( K \), number of components},
			ymin=1e-3,
			ymax=1.8*1e-3,
			ylabel= {Normalized MSE}, 
			ylabel shift = 0.0cm,
			grid = both,
			legend columns = 1,
			legend entries={
				\footnotesize \( M = 10 \),
				\footnotesize \( M = 50 \),
				\footnotesize \( M = 100 \),
				\footnotesize \( M = 200 \),
		    },
			legend style={at={(1.0,1.0)}, anchor=north east},
			legend cell align = {left},
			]
			
			\addplot[globalcov]
			table[x=comp, y=gmm_from_y_block, col sep=comma]
			{csvdat/GMM/plot_comp/2021-10-21_17-04-21_carrier=24_symbols=14_pilots=50_epochs=2_v=3kmh_ntrain=10000_snr=20.csv};
			
			\addplot[amp]
			table[x=comp, y=gmm_from_y_block, col sep=comma]
			{csvdat/GMM/plot_comp/2021-10-21_17-04-21_carrier=24_symbols=14_pilots=50_epochs=2_v=3kmh_ntrain=50000_snr=20.csv};
			
			\addplot[gmm]
			table[x=comp, y=gmm_from_y_block, col sep=comma]
			{csvdat/GMM/plot_comp/2021-10-21_17-04-21_carrier=24_symbols=14_pilots=50_epochs=2_v=3kmh_ntrain=100000_snr=20.csv};
			
			\addplot[gmmdiag]
			table[x=comp, y=gmm_from_y_block, col sep=comma]
			{csvdat/GMM/plot_comp/2021-10-21_17-04-21_carrier=24_symbols=14_pilots=50_epochs=2_v=3kmh_ntrain=200000_snr=20.csv};
			
		\end{axis}
	\end{tikzpicture}
	\caption{Wideband signal model (\Cref{sec:signal_model_wideband}) and QuaDRiGa channel model (\Cref{sec:quadriga}) with \( N_p = 50 \) pilot tones for \( N_c = 24 \) carriers and \( N_t = 14 \) time symbols for \( v = 3 \) km/h and block-type pilot pattern at an SNR of 20 dB. The GMM estimator is trained using \( M \cdot 10^3 \) samples.}
	\label{fig:wideband3}
\end{figure}

\section{Conclusion and Outlook}

We studied the behavior of a \ac{gmm} channel estimator when the number of \ac{gmm} components is increased.
The \ac{gmm} estimator is tailored for a particular
communications environment.
This is the case because the underlying \ac{gmm}
uses training data which stem, for example,
from the coverage area of a base station.
Thereafter, the \ac{gmm} estimator can be employed for channel estimation in this environment.

In case of an invertible observation matrix, we proved the convergence of the \ac{gmm} estimator to the optimal \ac{cme}.
Notably, the proof only assumes that a sequence of \acp{pdf} exists which converges uniformly to the true channel \ac{pdf}.
An example of such a sequence is given by \acp{gmm} but no properties unique to \acp{gmm} are used in the proof.
In particular, any sequence of \acp{cme} which is based on a uniformly convergent sequence of \acp{pdf} also converges to the optimal \ac{cme}.
This invites the study of new channel estimators based on other universal approximators.

While the theoretical results are of an asymptotic nature,
we analyzed a number of practically relevant settings (\ac{simo}, \ac{mimo}, and wideband) in numerical simulations.
There, already with a moderate number of \ac{gmm} components, the proposed estimator outperforms various state-of-the-art approaches in all depicted scenarios.
In particular, increasing the number of \ac{gmm} components lead to a performance improvement even for noninvertible observation matrices.
Additionally, we demonstrated how scenario-specific insights can be used to reduce the \ac{gmm} estimator's complexity.

\appendices

\section{Proof of Theorem 2}\label{sec:proof_of_main_result}

The proof of \Cref{thm:main_result} makes use of and is presented after the following lemma.

\begin{lemma}\label{lem:bound_norm}
For an arbitrary \( \mby \in \R^N \), it holds
\begin{multline*}
    \int \| \mbh \| f_{\Brv n}(\mby - \mbA\mbh) d \mbh
    \leq \\ \sqrt{\det(\mbA^{-1}\mbA^{-\tp})} \sqrt{\| \mbA^{-1} \mby \|^2 + \mathrm{trace}(\mbA^{-1}\mbSigma\mbA^{-\tp})}.
\end{multline*}
\end{lemma}
\begin{proof}
Recall that \( f_{\Brv n} \) denotes a Gaussian \ac{pdf} with mean zero and covariance matrix \( \mbSigma \).
We have
\begin{align}
    &f_{\mbn}(\mby - \mbA\mbh) = \frac{\exp(-\frac{1}{2}(\mby - \mbA\mbh)^\tp \mbSigma^{-1} (\mby - \mbA\mbh))}{\sqrt{(2\pi)^N \det(\mbSigma)}}
    \\&= \frac{\exp(-\frac{1}{2}(\mbA^{-1}\mby - \mbh)^\tp \mbA^\tp \mbSigma^{-1} \mbA (\mbA^{-1}\mby - \mbh))}{\sqrt{(2\pi)^N \det(\mbSigma)}}
    \\&= \frac{\exp(-\frac{1}{2}(\mbh - \mbA^{-1}\mby)^\tp (\mbA^{-1} \mbSigma \mbA^{-\tp})^{-1} (\mbh - \mbA^{-1}\mby))}{\sqrt{(2\pi)^N \det(\mbSigma)}}.
\end{align}
Therefore, \( (\sqrt{\det(\mbA^{-1} \mbA^{-\tp})})^{-1} f_{\mbn}(\mby - \mbA\mbh) =: \tilde{f}(\mbh) \) is a Gaussian \ac{pdf} with mean vector \( \tilde{\mbmu} := \mbA^{-1} \mby \) and covariance matrix \( \tilde{\mbSigma} := \mbA^{-1} \mbSigma \mbA^{-\tp} \).
We are interested in computing
\begin{align}
    \nonumber
    &\int \| \mbh \| f_{\Brv n}(\mby - \mbA\mbh) d \mbh
    \\&= \sqrt{\det(\mbA^{-1} \mbA^{-\tp})} \int \| \mbh \| \frac{f_{\mbn}(\mby - \mbA\mbh)}{\sqrt{\det(\mbA^{-1} \mbA^{-\tp})}} d \mbh
    \\&= \sqrt{\det(\mbA^{-1} \mbA^{-\tp})} \int \| \mbh \| \tilde{f}(\mbh) d \mbh
    \label{eq:to_be_bounded}
\end{align}
The integral in~\eqref{eq:to_be_bounded} computes the expected value of the norm of a random vector with the \ac{pdf} \( \tilde{f} \).
Let \( \Brv w \) be a standard Gaussian random vector (with mean \( \mbzero \in \R^N \) and covariance matrix \( \mbI \in \R^{N\times N} \)) and let \( \tilde{\mbSigma}^{\frac{1}{2}} \) be a square root of the covariance matrix \( \tilde{\mbSigma} = \tilde{\mbSigma}^{\frac{\tp}{2}} \tilde{\mbSigma}^{\frac{1}{2}} \).
Then, \( \tilde{\mbmu} + \tilde{\mbSigma}^{\frac{1}{2}} \Brv w \) is a random vector with the \ac{pdf} \( \tilde{f} \).
Thus, we can express the integral in~\eqref{eq:to_be_bounded} as:
\begin{equation}
    \int \| \mbh \| \tilde{f}(\mbh) d \mbh
     = \expec_{\Brv w}[\| \tilde{\mbmu} + \tilde{\mbSigma}^{\frac{1}{2}} \Brv w \|]
    \label{eq:expectation_wrt_w}
\end{equation}
where we compute the expected value with respect to the standard Gaussian random vector \( \Brv w \).
Jensen's inequality yields:
\begin{align}
    &\left( \expec_{\Brv w}[\| \tilde{\mbmu} + \tilde{\mbSigma}^{\frac{1}{2}} \Brv w \|] \right)^2
    \leq \expec_{\Brv w}[\| \tilde{\mbmu} + \tilde{\mbSigma}^{\frac{1}{2}} \Brv w \|^2]
    \\&= \| \tilde{\mbmu} \|^2 + 2 \tilde{\mbmu}^\tp \tilde{\mbSigma}^{\frac{1}{2}} \expec_{\Brv w}[\Brv w] + \expec_{\Brv w}[\mathrm{trace}(\Brv w^\tp \tilde{\mbSigma}^{\frac{\tp}{2}} \tilde{\mbSigma}^{\frac{1}{2}} \Brv w)]
    \\&= \| \tilde{\mbmu} \|^2 + \mathrm{trace}(\tilde{\mbSigma} \expec_{\Brv w}[\Brv w \Brv w^\tp])
    = \| \tilde{\mbmu} \|^2 + \mathrm{trace}(\tilde{\mbSigma})
    \\&= \| \mbA^{-1} \mby \|^2 + \mathrm{trace}(\mbA^{-1}\mbSigma\mbA^{-\tp}).
\end{align}
Now, we take the square root on both sides and plug the result into~\eqref{eq:expectation_wrt_w}.
Then, we use this in~\eqref{eq:to_be_bounded} to conclude.
\end{proof}

\begin{proof}[Proof of \Cref{thm:main_result}]
First, we show that the convergence of \( f_{\Brv h}^{(K)} \) to \( f_{\Brv h} \) implies the convergence of \( f_{\Brv y}^{(K)} \) to \( f_{\Brv y} \).
The \ac{pdf} of \( \Brv x := \mbA \Brv h \) is
\begin{equation}
    f_{\Brv x}(\mbx) = \frac{1}{|\det(\mbA)|} f_{\Brv h}(\mbA^{-1} \mbx)
\end{equation}
because \( \mbA \) is invertible.
Since the random vector \( \Brv y = \Brv x + \Brv n \) is a sum of two stochastically independent random vectors, its \ac{pdf} can be computed via convolution:
\begin{equation}\label{eq:convolution}
    f_{\Brv y}(\mby) = \int f_{\Brv n}(\mbs) f_{\Brv x}(\mby - \mbs) d\mbs.
\end{equation}
Similarly, \( f_{\Brv y}^{(K)} \) is obtained by replacing \( f_{\Brv x} \) with \( f_{\Brv x}^{(K)} \) in~\eqref{eq:convolution}.
For later reference, note that because \( f_{\Brv n} \) is positive (\( f_{\Brv n}(\mbs) > 0 \) for all \( \mbs \in \R^N \)) and \(f_{\Brv x}\) as well as \(f_{\Brv x}^{(K)}\) are continuous \acp{pdf}, the convolution results \( f_{\Brv y}^{(K)} \) and \( f_{\Brv y} \) are positive, too.
We have
\begin{align}
    &| f_{\Brv y}(\mby) - f_{\Brv y}^{(K)}(\mby) |
    \\&= \left| \int f_{\Brv n}(\mbs) \left(f_{\Brv x}(\mby - \mbs) - f_{\Brv x}^{(K)}(\mby - \mbs)\right) d\mbs \right|
    \\&\leq \int \left|f_{\Brv n}(\mbs) \frac{f_{\Brv h}(\mbA^{-1}(\mby - \mbs)) - f_{\Brv h}^{(K)}(\mbA^{-1}(\mby - \mbs))}{|\det(\mbA)|} \right| d\mbs
    \\&\leq \frac{\| f_{\Brv h} - f_{\Brv h}^{(K)} \|_\infty}{|\det(\mbA)|} \int | f_{\Brv n}(\mbs) | d\mbs = \frac{\| f_{\Brv h} - f_{\Brv h}^{(K)} \|_\infty}{|\det(\mbA)|}.
\end{align}
The last integral is equal to one because \( f_{\Brv n} \) is a \ac{pdf}.
Since \( \lim_{K\to\infty} \| f_{\Brv h} - f_{\Brv h}^{(K)} \|_\infty = 0 \) holds by assumption, we have
\begin{equation}\label{eq:convergence_fyK}
    \lim_{K\to\infty} \| f_{\Brv y} - f_{\Brv y}^{(K)} \|_\infty = 0.
\end{equation}

To show~\eqref{eq:convergence_estimator}, let \( \mby \in \calB_r \) be arbitrary.
With~\eqref{eq:conditional_mean} and~\eqref{eq:conditional_mean_K} in mind, we find the following upper bound:
\begin{multline}
    \| \hhat - \hhat^{(K)} \|
    \leq \int \| \mbh \| | f_{\Brv n}(\mby - \mbA\mbh) | \left| \frac{f_{\Brv h}(\mbh)}{f_{\Brv y}(\mby)} - \frac{f_{\Brv h}^{(K)}(\mbh)}{f_{\Brv y}^{(K)}(\mby)} \right| d \mbh
    \\\leq \sup_{\mbh\in\R^N} \left| \frac{f_{\Brv h}(\mbh)}{f_{\Brv y}(\mby)} - \frac{f_{\Brv h}^{(K)}(\mbh)}{f_{\Brv y}^{(K)}(\mby)} \right| \int \| \mbh \| f_{\Brv n}(\mby - \mbA\mbh) d \mbh.
\end{multline}
The last integral is independent of \( K \) and by \Cref{lem:bound_norm}, it is finite for any \( \mby \in \R^N \).
It is in particular bounded by
\begin{equation*}
    \sqrt{\det(\mbA^{-1}\mbA^{-\tp})} \sqrt{\| \mbA^{-1} \|^2 r^2 + \mathrm{trace}(\mbA^{-1}\mbSigma\mbA^{-\tp})} < \infty
\end{equation*}
for all \( \mby \in \calB_r \).
Hence, as soon as
\begin{equation}\label{eq:convergence_frac_xy}
    \lim_{K\to\infty} \sup_{\mbh\in\R^N} \left| \frac{f_{\Brv h}(\mbh)}{f_{\Brv y}(\mby)} - \frac{f_{\Brv h}^{(K)}(\mbh)}{f_{\Brv y}^{(K)}(\mby)} \right| = 0, \quad \forall \mby \in \calB_r
\end{equation}
is shown, \eqref{eq:convergence_estimator} is confirmed.
To prove~\eqref{eq:convergence_frac_xy}, we write
\begin{equation}\label{eq:frac_common_denominator}
   \left| \frac{f_{\Brv h}(\mbh)}{f_{\Brv y}(\mby)} - \frac{f_{\Brv h}^{(K)}(\mbh)}{f_{\Brv y}^{(K)}(\mby)} \right|
    =
    \left| \frac{f_{\Brv h}(\mbh) f_{\Brv y}^{(K)}(\mby) - f_{\Brv y}(\mby) f_{\Brv h}^{(K)}(\mbh)}{f_{\Brv y}(\mby) f_{\Brv y}^{(K)}(\mby)} \right|
\end{equation}
for an arbitrary \( \mbh \in \R^N \).
Now, we add \( 0 = f_{\Brv y}^{(K)}(\mby) f_{\Brv h}^{(K)}(\mbh) - f_{\Brv y}^{(K)}(\mby) f_{\Brv h}^{(K)}(\mbh) \) in the numerator on the right-hand side and apply the triangle inequality to get
\begin{align}
    &\left| \frac{f_{\Brv h}(\mbh) f_{\Brv y}^{(K)}(\mby) - f_{\Brv y}(\mby) f_{\Brv h}^{(K)}(\mbh) }{f_{\Brv y}(\mby) f_{\Brv y}^{(K)}(\mby)} \right|
    \\&\leq\nonumber
    \frac{\left| \left(f_{\Brv h}(\mbh) - f_{\Brv h}^{(K)}(\mbh)\right) f_{\Brv y}^{(K)}(\mby) \right|}{f_{\Brv y}(\mby) f_{\Brv y}^{(K)}(\mby)}
    \\&+
    \frac{\left| \left(f_{\Brv y}^{(K)}(\mby) - f_{\Brv y}(\mby)\right) f_{\Brv h}^{(K)}(\mbh) \right|}{f_{\Brv y}(\mby) f_{\Brv y}^{(K)}(\mby)}
    \\&\leq \label{eq:frac_upper_bound}
    \frac{\| f_{\Brv h} - f_{\Brv h}^{(K)} \|_\infty \| f_{\Brv y}^{(K)} \|_\infty
    + \| f_{\Brv y}^{(K)} - f_{\Brv y} \|_\infty \| f_{\Brv h}^{(K)} \|_\infty}{f_{\Brv y}(\mby) f_{\Brv y}^{(K)}(\mby)}.
\end{align}

By the compactness of \( \calB_r \) and continuity of \( f_{\Brv y} \), there exists a \( \mby_{\text{min}} \in \calB_r \) at which \( f_{\Brv y} \) attains a minimum value \( f_{\Brv y}(\mby_{\text{min}}) > 0 \) over \( \calB_r \).
Due to the uniform convergence~\eqref{eq:convergence_fyK}, there exists an index \( N_1 \in \N \) such that \( | f_{\Brv y}(\mby) - f_{\Brv y}^{(K)}(\mby)| \leq \frac{1}{2} f_{\Brv y}(\mby_{\text{min}}) \) holds for all \( K \geq N_1 \) and for all \( \mby \in \calB_r \).
The reverse triangle inequality then shows that \( f_{\Brv y}^{(K)}(\mby) \geq f_{\Brv y}(\mby) - | f_{\Brv y}(\mby) - f_{\Brv y}^{(K)}(\mby) | \geq f_{\Brv y}(\mby_{\text{min}}) - \frac{1}{2} f_{\Brv y}(\mby_{\text{min}}) = \frac{1}{2} f_{\Brv y}(\mby_{\text{min}}) \) is true.
Hence, with \( M_1 := \frac{1}{2} f_{\Brv y}(\mby_{\text{min}}) > 0 \), the inequality
\begin{equation}\label{eq:N1}
    f_{\Brv y}^{(K)}(\mby) \geq M_1
\end{equation}
holds for all \( \mby \in \calB_r \) and for all \( K \geq N_1 \).
Further, since \( \| f_{\Brv h} - f_{\Brv h}^{(K)} \|_\infty \to 0 \) and \( \| f_{\Brv h} \|_\infty < \infty \), there exist \( M_2 > 0 \) and \( N_2 \in \N \) such that
\begin{equation}\label{eq:N2}
    \| f_{\Brv h}^{(K)} \|_\infty \leq M_2 \quad\text{for all}\quad K \geq N_2.
\end{equation}
Analogously, there exist \( M_3 > 0 \) and \( N_3 \in \N \) such that
\begin{equation}\label{eq:N3}
    \| f_{\Brv y}^{(K)} \|_\infty \leq M_3 \quad\text{for all}\quad K \geq N_3.
\end{equation}

Let \( \varepsilon > 0 \) be arbitrary.
Due to \( \| f_{\Brv h} - f_{\Brv h}^{(K)} \|_\infty \to 0 \), there exists an index \( N_4 \geq \max\{N_1, N_3\} \) such that
\begin{equation}\label{eq:N4}
    \| f_{\Brv h} - f_{\Brv h}^{(K)} \|_\infty \leq \frac{f_{\Brv y}(\mby_{\text{min}}) M_1}{2 M_3} \varepsilon \quad\text{for all}\quad K \geq N_4.
\end{equation}
Similarly, there exists an index \( N_5 \geq \max\{N_1, N_2\} \) with
\begin{equation}\label{eq:N5}
    \| f_{\Brv y} - f_{\Brv y}^{(K)} \|_\infty \leq \frac{f_{\Brv y}(\mby_{\text{min}}) M_1}{2 M_2} \varepsilon \quad\text{for all}\quad K \geq N_5.
\end{equation}
We can use the last five inequalities to bound~\eqref{eq:frac_upper_bound}.
To this end, \( f_{\Brv y}(\mby_{\text{min}}) \) and \eqref{eq:N1} provide bounds on the terms in the denominator, \eqref{eq:N4} and \eqref{eq:N3} bound the first summand in the numerator, and \eqref{eq:N5} and \eqref{eq:N2} bound the second summand in the numerator.
In total, this yields an upper bound on~\eqref{eq:frac_common_denominator}:
\begin{align}
    \nonumber
    \left| \frac{f_{\Brv h}(\mbh)}{f_{\Brv y}(\mby)} - \frac{f_{\Brv h}^{(K)}(\mbh)}{f_{\Brv y}^{(K)}(\mby)} \right|
    &\leq \frac{f_{\Brv y}(\mby_{\text{min}}) M_1}{2 M_3} \varepsilon \cdot \frac{M_3}{f_{\Brv y}(\mby_{\text{min}}) M_1} \\&+ \frac{f_{\Brv y}(\mby_{\text{min}}) M_1}{2 M_2} \varepsilon \cdot \frac{M_2}{f_{\Brv y}(\mby_{\text{min}}) M_1}.
\end{align}
for all \( K \geq \max\{N_4, N_5\} \) and for all \( \mby \in \calB_r \).
We conclude
\begin{equation}
    \sup_{\mbh\in\R^N} \left| \frac{f_{\Brv h}(\mbh)}{f_{\Brv y}(\mby)} - \frac{f_{\Brv h}^{(K)}(\mbh)}{f_{\Brv y}^{(K)}(\mby)} \right|
    \leq \varepsilon
\end{equation}
for all \( K \geq \max\{N_4, N_5\} \) and all \( \mby \in \calB_r \), and because \( \varepsilon \) was arbitrary, \eqref{eq:convergence_frac_xy} is confirmed, which finishes the proof.
\end{proof}

\section{Integral for Noninvertible Matrices}\label{sec:noninvertible_A}

To see why \( \int \| \mbh \| f_{\mbn}(\mby - \mbA \mbh) d\mbh \) might not be finite, consider the matrix \( \mbA = [\mbI, \mbzero] \in \R^{m\times N} \) where \( \mbI \in \R^{m\times m} \) is the identity matrix and the remaining matrix elements are zero.
Let us write \( \mbh = [\mbh_m^\tp, \mbh_{N-m}^\tp]^\tp \in \R^m \times \R^{N-m} \).
Define the set \( \calC = \{ \mbh \in \R^N \mid \| \mbh_m \| \geq 1, \| \mbh_{N-m} \| \geq 1 \} \) where the norm of both sub-vectors \( \mbh_m \) and \( \mbh_{N-m} \) is at least one such that we always have \( \| \mbh \| \geq 1 \) on \( \calC \).
We can now compute:
\begin{align}
    &\int_{\R^N} \| \mbh \| f_{\mbn}(\mby - \mbA \mbh) d\mbh
    \geq \int_{\calC} f_{\mbn}(\mby - \mbA \mbh) d\mbh
    \\&= \int_{\| \mbh_{N-m} \| \geq 1} \int_{\| \mbh_m \| \geq 1} f_{\mbn}(\mby - \mbh_m) d\mbh_m d\mbh_{N-m}.
\end{align}
Since \( f_n \) is an \( m \)-dimensional Gaussian \ac{pdf}, the inner integral is equal to some constant \( c \) with \( 0 < c < 1 \) and it follows that \( \int \| \mbh \| f_{\mbn}(\mby - \mbA \mbh) d\mbh \) is not finite.

\section{On the Uniform Convergence}\label{sec:no_uniform_convergence}

Let us express the \ac{pdf} of \( \mbA \mbh \).
If \( \mbA \in \R^{m\times N} \) is a wide matrix with full rank \( m \), we can assume that the first \( m \) columns are linearly independent (otherwise we introduce a permutation matrix).
This allows us to partition \( \mbA = [\mbA_i, \mbA_n] \) into an invertible part \( \mbA_i \in \R^{m\times m} \)
and a noninvertible part \( \mbA_n \in \R^{m\times N-m} \).
With a corresponding partitioning of \( \mbh = [\mbh_i^\tp, \mbh_n^\tp]^\tp \in \R^m \times \R^{N-m} \),
we have \( \mbx = \mbA \mbh = \mbA_i \mbh_i + \mbA_n \mbh_n \) and
we can define an invertible mapping \( t: \R^N \to \R^N \):
\begin{align}
    t: (\mbh_i, \mbh_n) &\mapsto (\mbA_i \mbh_i + \mbA_n \mbh_n, \mbh_n) = (\mbx, \mbx')
    \\t^{-1}: (\mbx, \mbx') &\mapsto (\mbA_i^{-1}(\mbx - \mbA_n \mbx'), \mbx') = (\mbh_i, \mbh_n).
\end{align}
We can now compute the joint density \( f_{\mbx,\mbx'} \) with the usual transformation formula:
\begin{equation}
    f_{\mbx, \mbx'}(\mbx, \mbx') =
    \frac{f_{\mbh}(t^{-1}(\mbx, \mbx'))}{\left|\det\left(\frac{\partial t}{\partial \mbh}(t^{-1}(\mbx, \mbx'))\right)\right|}.
\end{equation}
Together with \( \left|\det\left(\frac{\partial t}{\partial \mbh}(t^{-1}(\mbx, \mbx'))\right)\right| = | \det(\mbA_i) | \),
we can express the \ac{pdf} of \( \mbx = \mbA \mbh \) via marginalization:
\begin{equation}
    f_{\mbx}(\mbx)
    = \int_{\R^{N-m}} f_{\mbx,\mbx'}(\mbx,\mbx') d \mbx'
    = \int_{\R^{N-m}} \frac{f_{\mbh}(t^{-1}(\mbx, \mbx'))}{| \det(\mbA_i) |} d \mbx'.
\end{equation}
Analogously, one obtains \( f_{\mbx}^{(K)} \) for \( \mbx^{(K)} = \mbA \mbh^{(K)} \).
Given
\begin{multline}
    |f_{\mbx}^{(K)}(\mbx) - f_{\mbx}(\mbx)| = \frac{1}{| \det(\mbA_i) |} \times \\ \left| \int_{\R^{N-m}} \left(  f_{\mbh}^{(K)}(t^{-1}(\mbx, \mbx')) - f_{\mbh}(t^{-1}(\mbx, \mbx')) \right) d \mbx' \right|,
\end{multline}
we can conjecture that due to the integral over \( \R^{N-m} \) the uniform convergence of \( f_{\mbh}^{(K)} \) to \( f_{\mbh} \) alone is generally not sufficient to infer the uniform convergence of \( f_{\mbx}^{(K)} \) to \( f_{\mbx} \).

\bibliographystyle{IEEEtran}
\bibliography{IEEEabrv,references}

\end{document}